\documentclass[a4paper,11pt]{article}
\usepackage{graphicx}
\usepackage{fullpage}
\usepackage[utf8]{inputenc}
\usepackage{array}
\newcolumntype{P}[1]{>{\centering\arraybackslash}p{#1}}
\usepackage{wrapfig}
\usepackage{float}
\usepackage{multirow}
\usepackage[export]{adjustbox}
\usepackage{tabularx}
\usepackage{color}
\usepackage[linesnumbered,ruled,vlined,noend]{algorithm2e}
\usepackage{caption}
\usepackage[english]{babel}
\usepackage{amsthm}
\usepackage{comment}
\usepackage{amsmath}
\usepackage{listings}
\usepackage{xcolor}
\usepackage{amsbsy}
\usepackage{todonotes}
\usepackage{longtable}
\usepackage{subcaption}

\definecolor{codegreen}{rgb}{0,0.6,0}
\definecolor{codegray}{rgb}{0.5,0.5,0.5}
\definecolor{codepurple}{rgb}{0.58,0,0.82}
\definecolor{backcolour}{rgb}{0.95,0.95,0.92}

\lstdefinestyle{mystyle}{
  backgroundcolor=\color{backcolour},   commentstyle=\color{codegreen},
  keywordstyle=\color{magenta},
  numberstyle=\tiny\color{codegray},
  stringstyle=\color{codepurple},
  basicstyle=\ttfamily\footnotesize,
  breakatwhitespace=false,
  breaklines=true,
  captionpos=b,
  keepspaces=true,
  numbers=left,
  numbersep=5pt,
  showspaces=false,
  showstringspaces=false,
  showtabs=false,
  tabsize=2
}

\usepackage[inline]{enumitem}
\newcommand{\RNum}[1]{\uppercase\expandafter{\romannumeral #1\relax}}
\newtheorem{theorem}{Theorem}
\usepackage{blindtext}
\usepackage{amssymb}
\usepackage{mathtools}
\newtheorem{lemma}[theorem]{Lemma}
\newtheorem{proposition}[theorem]{Proposition}

\newtheorem{definition}{Definition}

\SetCommentSty{mycommfont}

\newcommand{\br}[1]{\left\lbrace #1 \right\rbrace}

\newcommand{\say}[1]{``#1''}

\newcommand{\eps}{\varepsilon}

\begin{document}
\title{Online bin packing of squares and cubes}
\author{Leah Epstein\thanks{Department of Mathematics, University of Haifa, Haifa, Israel.  \texttt{lea@math.haifa.ac.il}} \and Loay Mualem\thanks{Department of Computer science, University of Haifa, Haifa, Israel.  \texttt{loaymua@gmail.com}.}}

\date{}

\maketitle

\begin{abstract}
In the $d$-dimensional online bin packing problem, $d$-dimensional cubes of positive sizes no larger than $1$ are presented one by one to be assigned to positions in $d$-dimensional unit cube bins. In this work, we provide  improved upper bounds on the asymptotic competitive ratio for square and cube bin packing problems, where our bounds do not exceed $2.0885$ and $2.5735$ for square and cube packing, respectively.  To achieve these results, we adapt and improve a previously designed harmonic-type algorithm, and apply a different method for defining weight functions. We detect deficiencies in the state-of-the-art results by providing counter-examples to the current best algorithms and the analysis, where the claimed bounds were $2.1187$ for square packing and $2.6161$ for cube packing.
\end{abstract}

\section{Introduction}
\textit{Bin Packing}  (BP) has been the cornerstone of approximation algorithms and has been extensively studied since the early 1970's. This problem and its variants are important problems with numerous classic applications, such as machine scheduling, cutting stock problems, and storage allocation. Recent applications include also cloud storage.

Bin packing was first introduced and investigated by Ullman in $1971$ \cite{U} (see also \cite{BBDEL_ESA18,BBDEL_newlb,B2,C,C1,FVL81,P11,J1,J2,KK82,S1,V}).
In the classic or standard one-dimensional bin packing problem, we are given a list
$L = \{i_1, i_2, \dots,i_n\}$ of items, and item sizes $ S = \{s_1, s_2, \dots, s_n\}$,
where $s_j \in (0,1]$ is the size of $i_j$ for any $1 \leq j \leq n $.
The goal is to pack these items into the minimum number of bins for this input.
More precisely, for a subset of items $B$, we let $|B|=\sum_{i_j \in B}s_j$, and
the goal is to partition $L$ into a set of subsets $B=\{b_1,b_2,b_3, \dots,b_{\ell} \}$,
where $1 \leq \ell \leq n$, such
that $|b_k| \leq 1$ holds for $k = 1, \dots, \ell$, where $\ell$ is minimized.

A bin packing algorithm can belong to one of two classes,  {\it online} or {\it offline}.
A bin packing algorithm is called
{\it online}  if it is given the items from $L$ one at a time, and it
must assign each item into a bin immediately upon arrival. A newly
arriving item  is packed according to the packing and sizes of
items that have already been presented before its arrival. There is no information
about subsequent items, and removing an already packed item from its position is not allowed.

As opposed to online algorithms,
offline algorithms for bin packing have complete knowledge about
the list of items. An offline algorithm simply maps $L$ into a set of bins (in a valid way), and the
ordering of the items in $L$ plays no role.

The offline problem is known to be NP-hard \cite{G}; thus,  research for this variant has
concentrated on the study and development of fast algorithms that
can produce near-optimal solutions for the problem in polynomial time.
That is, extensive  research has gone into developing approximation algorithms for this problem.
These algorithms have proven performance for any possible input, and process the input items in polynomial time. See \cite{FVL81,KK82,Rothv,SL94} for such work.

Online algorithms are analyzed via the (absolute or asymptotic) competitive ratio. This is the worst-case cost ratio between outputs of an online algorithm and those of an optimal offline algorithm (for the same inputs). There is vast research on online variants as well \cite{BBDEL_ESA18,BBDEL_newlb,B2,C,C1,P11,J1,J2,S1,V}.

We define the competitive ratio
more precisely. Given an input list $L$,
let $ALG(L)$ be the cost (number of bins used) obtained by applying
algorithm $ALG$ on the input $L$.
Let $OPT$ be an optimal offline algorithm, that uses the minimum number
of bins for packing the items, and let $OPT(L)$ denote the number of bins that $OPT$ uses for a given
input $L$. The algorithm is \textit{absolutely r-competitive} if for any input
$ALG(L)\leq r\cdot OPT(L)$  and \textit{asymptotically}
$r$-competitive  if there exists a constant $C$ such that for any input $ALG(L) \leq r\cdot OPT(L) + C$. The asymptotic competitive ratio for $ALG$ is the infimum $r$ such that $ALG$ is asymptotically $r$-competitive. Since the last measure is the common one for bin packing, we only discuss this measure in this text, and sometimes omit the words asymptotic and asymptotically. For offline problems, the approximation ratio is defined analogously.

In this work, we deal with {\it online bin packing of cubes}, and we improve the asymptotic competitive ratio for the $d$-dimensional bin packing problem of cubes for $d=2$ and $d=3$. In the $d=2$, cubes are in fact squares, as can be seen in Figure~\ref{fig1}. The figure contains an example of an optimal solution for square packing ($d=2$), where the input consists of one item with side $0.5$, two items with side $\frac 13$, ten items with side $\frac 16$, and four items with side $\frac 14$. The case $d=1$ is sometimes seen as the classic variant of bin packing.

We define the more general case of box packing as follows. The input consists of  a list \textit{L} of items, where each
item is a \textit{d}-dimensional box, and in each dimension, the side length of an item does not exceed $1$. The output is a packing of all input items of \textit{L} into $d$-dimensional hyper-cube bins. The goal is to
 minimize the number of used bins. A packing is an assignment of positions in bins to all items such that the following two requirements hold. No two items in a bin overlap with each other (except for their boundaries), and
the sides of item are parallel to sides of bins.
Note that we do not exclude the option of rotation, also called non-oriented box packing, though we deal here with the asymptotic competitive ratio for squares and cubes, where rotation is meaningless.

\begin{figure}[h!]
\centering
\captionsetup{justification=centering,margin=4cm}

\includegraphics[scale=.20]{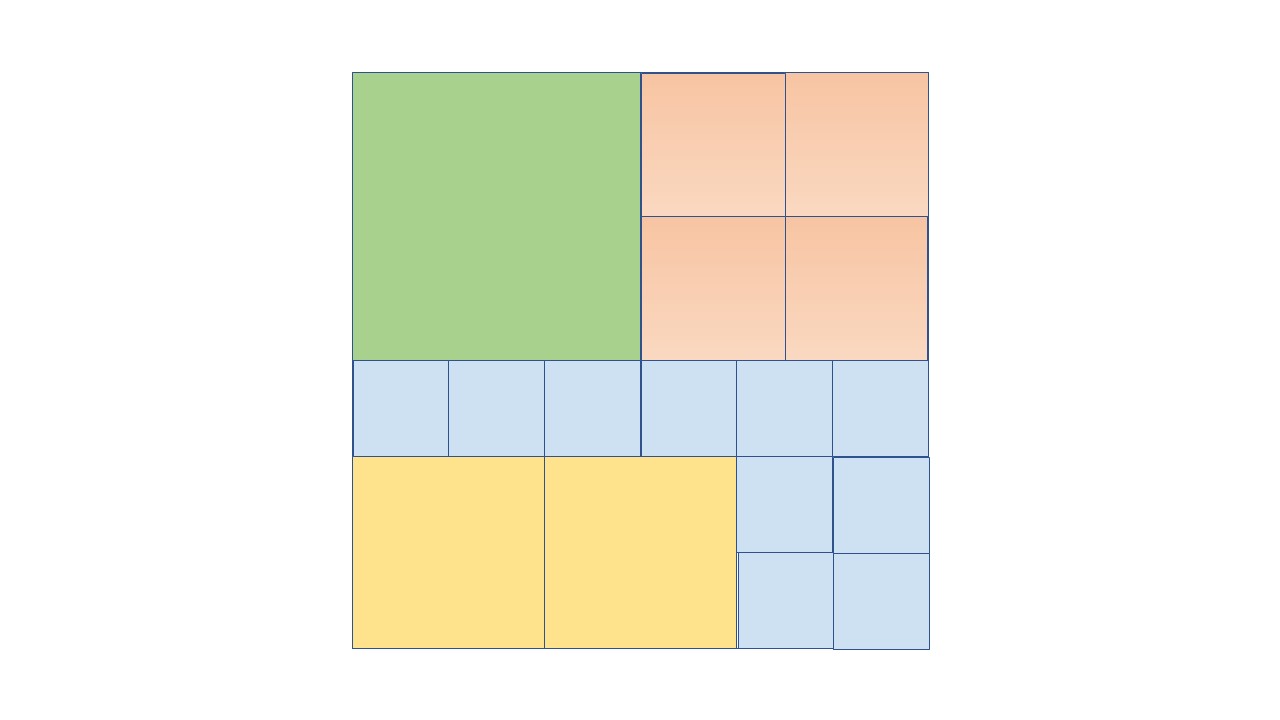}

\caption{An example for square packing. \label{fig1}}
\end{figure}

\subsection{Previous results}
Recall that bin packing is an NP-hard problem, and a large part of the research in this field of study focused on finding
(asymptotic) approximation bounds. The offline problem has asymptotic approximation schemes (where an approximation scheme is a family of asymptotic approximation algorithms with approximation ratio $1+\varepsilon$ for any $\varepsilon>0$) \cite{FVL81,KK82}.
The online bin packing problem was first introduced by Ullman~\cite{U}. Johnson~\cite{J1} showed that a greedy algorithm called Next Fit (NF) (defined below) has an asymptotic (and an absolute) competitive ratio of 2. It was also shown by Johnson et al.~\cite{J2} that another greedy algorithm called  First Fit has an asymptotic competitive ratio of $\frac{17}{10}$~\cite{J2}. As for lower bounds, Yao showed that no online algorithm has an competitive ratio smaller than $1.5$~\cite{Yao}. Later, Brown~\cite{brown1979lower} and Liang~\cite{liang1980lower} improved this lower bound to $1.53635$, and Van Vliet~\cite{V} improved this lower bound known to $1.54014$. Balogh et al.~\cite{B2} improved this lower bound to $1.54037$.  The tightest lower bound known so far is $1.54278$ By Balogh et al.~\cite{BBDEL_newlb}.
As for lower bounds, Yao showed that no online algorithm has an competitive ratio smaller than $1.5$~\cite{Yao}. Later, Brown~\cite{brown1979lower} and Liang~\cite{liang1980lower} improved this lower bound to $1.53635$, and Van Vliet~\cite{V} improved this lower bound known to $1.54014$. Balogh et al.~\cite{B2} improved this lower bound to $1.54037$.  The tightest lower bound known so far is $1.54278$ by Balogh et al.~\cite{BBDEL_newlb}.

Lee and Lee~\cite{LeeLee85} presented the Harmonic algorithm (see below).
This algorithm uses bounded space (a constant number of bins can receive items at each time), and it achieves an asymptotic competitive ratio of approximately $1.69103$ for large values of its parameter. They also developed the Refined Harmonic algorithm, which has an asymptotic competitive ratio that does not exceed $1.63597$. Shortly afterwards, Ramanan et al. ~\cite{RaBrLL89} introduced  Modified Harmonic and Modified Harmonic 2, and showed that these algorithms have asymptotic competitive ratios not exceeding approximately $1.61562$ and $1.61217$, respectively. The upper bound was improved further later~\cite{S1,P11}. The best upper bound on the asymptotic competitive ratio known so far is $1.57829$ by  Balogh et al.~\cite{BBDEL_ESA18}.

In this work, we study the $d$-dimensional bin packing of cubes for $d=2,3$, and improve the existing bounds for this problem. In what follows, we discuss previous work for that variant.
The hyper-cube online packing problem was studied by Coppersmith and Raghavan~\cite{CopRag89} who showed upper bounds of $2.6875$ and $6.25$ on the asymptotic competitive ratios for online square packing and online cube packing, respectively.
Seiden and van Stee~\cite{SeidenS03} improved the upper bound for square packing to $395/162\approx 2.438272$, Miyazawa and Wakabayashi~\cite{MiyWak03} improved the upper bound for cube packing to $3.954$, where the algorithm was based on that of~\cite{CopRag89}.
Epstein and van Stee proved an upper bound of $2.24437$ for square packing and
an upper bound of $2.9421$ for online cube packing ~\cite{epstein2005online}. These algorithms are similar to the one-dimensional modifications of harmonic algorithms.
Han et al.~\cite{HanYZ10} gave  upper bounds of $2.1187$ and $2.6161$ for the asymptotic competitive ratios for square packing and cube packing, respectively. We note that these last bounds are not valid for the algorithms as they were defined and their analysis, as we show in this work, by providing counter-examples to the action of the algorithm for $d=2$, and by explaining why the analysis does not hold in general (see Section \ref{55}).
There is also an earlier version of that work~\cite{HDZarxiv} but there are flaws in that analysis as well.
As for lower bounds on the competitive ratio, there has been some work on that direction as well ~\cite{SeidenS03,epstein2005online,vanv,HeyS17,balogh2017lower}, and the current best lower bound is approximately 1.75154~\cite{balogh2017lower}.

The offline variant of the square and cube bin packing also have asymptotic approximation schemes~\cite{BCKS06}.
In addition, there is work for more general variants of online rectangle and box bin packing with or without rotation, see~\cite{Blitz,vanv,CopRag89,HeyS17,csirik1993two,Csivan93,epstein2005optimal,Ep10,fujita2002two,G91,GV94,han2011new,vanv,Ep19},  for other variants of multi-dimensional packing (see the survey \cite{CKPT16}), and for $d$-dimensional vectors (see for example \cite{ACKS13}). Naturally, the bounds for the more general case are larger.

\subsection{Our contribution}
For decades, bounds for asymptotic competitive ratios of bin packing problems have been extensively studied. In this work, we present improved results for $d$-dimensional bin packing problem:
\begin{itemize}
\item We provide a new harmonic-type algorithm for $d$-dimensional bin packing problem. The key components of our algorithm are classification of the items and extension of the framework suggested by~\cite{LeeLee85} with respect to the one-dimensional bin packing problem. Our algorithm is specified by a general structure that is based on that of \cite{HanYZ10} and a new set of parameters.
\item We provide a new weighting technique for the $d$-dimensional bin packing problem. A related method was used in the past for standard bin packing \cite{BBDEL_ESA18}, that is, for the one-dimensional case,
but no such method was defined for variants in multiple dimensions. Here we show that it allows one to improve the bounds for another bin packing problem.
\item To emphasize the effectiveness of our new suggested algorithm, we established an improved asymptotic competitive ratio for the cases $d=2,3$. We obtain tighter upper bounds for the asymptotic competitive ratio for the online square and cube packing. Specifically, the algorithms have asymptotic competitive ratios of at most $2.0885$ and $2.5735$, respectively. This is to be compared to the currently known bounds of $2.1187$ for square packing and $2.6161$ for cube packing by~\cite{HanYZ10} (which are unfortunately incorrect, but it might be possible to prove slightly inferior bounds for these algorithms using our method of analysis).
\item We present a counter example for the previous upper bound claimed by~\cite{HanYZ10} for $d=2$, showing that it is higher than $2.12$. We also explain why their analysis is incorrect and cannot yield a bound below $2.24$ for square packing (though we believe that an upper bound of approximately $2.14$ can be shown for their algorithm for $d=2$ using our method of analysis).
\end{itemize}

Our analysis is based on introducing weight functions and bounding the asymptotic competitive ratio by showing that the total weight of bins of the algorithm is equal to the total weight (up to an additive constant) while bounding the total weight of any bin of an optimal offline solution from above \cite{LeeLee85,S1}. For obtaining the upper bounds on weights, we use computer-assisted proofs.

The organization of this work is as follows. In Section \ref{22} we present the harmonic algorithm and the algorithm for packing small items. In Section~\ref{33} we present our algorithm for the online $d$-dimensional bin packing problem of squares and cubes.
In Section~\ref{44}, we present our weighting functions and present their analysis for our algorithm and for optimal solutions.  Section~\ref{66} contains the specific parameters for our algorithms, which lead to the improved bounds.
Finally, in Section~\ref{55} we show the counter example for the algorithm of~\cite{HanYZ10}.

\section{Preliminaries}
\label{22}
In this section we provide the details of algorithms that will be used in our work. We will use these definitions later, and additional definitions will be introduced as required.

\subsection{Harmonic algorithms}
We start with the definition of NF, which is an algorithm for one dimensional bin packing.  NF keeps one bin open at a time.
If the next item fits in the current open bin, that is, packing it into the bin keeps the validity of the packing, the new item is packed there.
Otherwise, the current open bin is closed, a new bin is opened, and the new item is packed there.

The main idea of harmonic algorithms is classifying each item by its size to a type, and packing it according to its type (as opposed to its exact size) in the following way:
Given a list $S=(\ell_1,\ell_2,\ldots,\ell_n)$, where $\ell_i \in (0,1]$ for all $i$, the classification of items is done according to partitioning the unit interval
$(0,1]$ to $M$ disjoint sub-intervals of the $I_j=(1/(j+1),1/j]$ form for $j = 1,\ldots,M-1$, and the $M$th sub-interval is $(0,1/M]$.
 Each item $\ell_i$ is called an item of type $j$ if $\ell_i\in I_j$, for some $1\leq j\leq M$.
 Similarly, each bin is classified. A bin is called {\it a bin of type $i$}, if it is designated to pack items of type $i$ exclusively.
Note that, each  bin of type $i$ for $1\leq i\leq M$ can accommodate at most $i$ items of type  $i$. Every type of items is packed independently of other types using NF.
A bin is called \say{active} if it still receive items, where otherwise it is considered to be \say{closed}, as in the definition of NF.

The harmonic algorithm has multiple advantages. It is efficient in terms of time and space complexity,  and the output is mostly independent of the exact arrival order of items.
On the other hand, a crucial disadvantage of this algorithm is that any item $\ell_i \in I_1$, that is, item with size larger than $1/2$, is packed into a bin alone, possibly wasting a large amount of space in this bin. A similar situation is encountered for $I_2$ as well.
Due to this drawback, other algorithms were suggested~\cite{LeeLee85,RaBrLL89,S1,P11,BBDEL_ESA18}. Some of these algorithms still classify items into types, but they mix two types in one bin, and this allows them to save space.

The weighting function used for the analysis of the basic harmonic algorithm is simple, the weight of every item of $I_i$ for $i<M$ is simply $\frac 1i$, and for small items, which are the items of type $M$, the weight is just slightly larger than the size \cite{LeeLee85,Woegin93}. Modifications of harmonic algorithms require more complicated weight functions, and different functions for different kinds of outputs (or inputs) \cite{LeeLee85,RaBrLL89}. Seiden \cite{S1} generalized the concept of weighting functions to weighting systems, but the necessity of this concept is still unclear.

The harmonic approach (without combining different items types into one bin) was used for the multi-dimensional case too \cite{Csivan93,epstein2005optimal,epstein2007bounds}, where the bounds are obviously higher than those of arbitrary online algorithms. For example, the tight bound on the asymptotic competitive ratio for squares is in $(2.36,2.37)$ \cite{epstein2007bounds}. Generalization were also used for squares and cubes \cite{epstein2005online,HanYZ10}.





The one-dimensional generalized  harmonic algorithms are instances of a general class of algorithms called \emph{Super Harmonic}~\cite{S1}. These algorithms are also interval classification algorithms, the main differences between these type of algorithms and the \emph{Harmonic} algorithm are as follows.  The intervals are predefined as before, but they are arbitrary and more general. Let $t_1 = 1 > t_2 > \cdots > t_{n+1}>0$ be rational numbers, the interval $I_j$ is defined to be $(t_{j+1},t_j]$ for $j = 1,\ldots,n+1$, where in the \emph{Harmonic} algorithm, the interval $I_j$ is defined to be $I_j=(1/(j+1),1/j]$.
Each item will be colored red or blue, and the packing is different for the items of different colors for every class. The partition of items of one class into two sub-types allows one to (partially) overcome the wasted space disadvantage in the harmonic algorithm. The methods for doing this are discussed later.
In the next section,
we modify and extend the \emph{Super Harmonic} algorithm to online hyper-cube packing. This was already done \cite{HanYZ10}, and we present this approach for completeness.
We also modify the weighting method in order to obtain an improved asymptotic competitive ratio for square and cube packing. We do not use the approach of \cite{HanYZ10}, but an approach previously used for the one-dimensional case \cite{BBDEL_ESA18}. The method is strongly related to that of Seiden \cite{S1}, but it is presented and used in a much simpler way. The (corrected) approach of \cite{HanYZ10} is a special case of our method that can lead to much higher uppers bounds on the asymptotic competitive ratio for each algorithm.


\subsection{Hyper-cube packing: packing \emph{small} items}
In what follows,  we will borrow some of the notation used in the context of \emph{Super Harmonic} algorithms. For a hyper-cube $h$, we use $s(h)$ to denote its side length, and we call it {\it size} or {\it side}.
We categorize our items into two types, \emph{large} and \emph{small}, where our algorithm handles each type differently; see Algorithm~\ref{EH}. Let $M$ be a fixed positive integer. An hyper-cube $t$ is defined to be a \emph{large} item if $s(t) \geq 1/M$, where otherwise it will categorized as \emph{small}. Although each item will be further categorized (assigned type) depending on its size,  \emph{large} and \emph{small} items can not share the same type (category).

In this framework, we use Algorithm AssignSmall from~\cite{epstein2005optimal} to pack small items.
Consider an item $t$ of size $s(t) \leq 1/M$, let $k$ be the largest non-negative integer such that $2^k\cdot s(t) \leq 1/M$, and let $i$ be
the integer such that $2^k\cdot s(t)\in\left(\frac{1}{i+1},\frac{1}{i}\right]$ where $i\in\br{M,\ldots,2M-1}$. Then, item $t$ is classified as type $i$ item.
The key ideas of the AssignSmall algorithm are:
\begin{enumerate}
\item For every $i \in \br{M, \ldots, 2M-1}$, there is at most one \say{active} bin which contains an item of type $i$, i.e., there will be at most $M$ \say{active} bins.
\item Each bin may be partitioned into multiple sub-bins which are hyper-cubes of different size of the form $1/(2^j\cdot i)$.
\end{enumerate}

We now introduce algorithm AssignSmall for online hyper-cube packing by~\cite{epstein2005optimal}. The algorithm AssignSmall handles the item $t$ of type $i$ as follows:
\begin{enumerate}
\item If there is an empty sub-bin of size $1/(2^k\cdot i)$, then the item is simply assigned there and placed anywhere within the sub-bin.
\item Else, if there is no empty sub-bin of any size $1/(2^j\cdot i)$ for $j<k$ inside the current bin, the bin is closed and a new bin is opened and partitioned into sub-bins of size $1/i.$ Then the procedure
in step $3$ is followed, or step $1$ in case $k=0$.
\item Take an empty sub-bin of size $1/(2^j\cdot i)$ for a maximum $j<k$. Partition it into $2^d$ identical
sub-bins (by cutting into two identical pieces, in each dimension). If the resulting sub-bins are larger than $1/(2^k\cdot i)$, take $on$ of them and partition it in the same
way. This is done until sub-bins of size $1/(2^k\cdot i)$ are reached. The new item is assigned into one such sub-bin.
\end{enumerate}
The following Lemma is taken directly from~\cite{epstein2005optimal}, and used also in \cite{HanYZ10}.
\begin{lemma}
\label{leahsmall}
In the AssignSmall algorithm, In each closed bin of type $i \geq M$, the occupied volume is at least
$\frac{(i^d-1)}{(i+1)^d} \geq \frac{(M^d-1)}{(M+1)^d}$.
\end{lemma}

\section{Algorithm Extended Harmonic (EH)}\label{33}
In this section, we define our algorithm Extended Harmonic (EH) for hyper-cube packing.
For any $M \geq 110$, Let $N$ be fixed positive integer and let $t_i \in [0,1]$ for every $i \in \br{1,\ldots, N+1}$ such that $t_i \geq t_{i+1}$, $t_1 = 1$, $t_{N+1} = 1/M$. We also define the interval $I_i$ to be $\left( t_{i+1}, t_i \right]$ for every $i \in \br{1, \ldots, N}$. An item $t$ is categorized as type $i$ if it falls in the interval $I_i$, i.e., $s(t)\in \left(t_{i+1},t_i \right]$.


The algorithm is split into two main components,  where we categorize all items into \emph{small} and \emph{large}, small items are packed by AssignSmall, and large items are packed by EH which we defined in this section.


For every $i \in \br{1, \ldots, N}$, each item of type $i$, is either colored red or blue. We then define two sets of counters $\br{e_j}_{j=1}^N$ and $\br{n_j}_{j=1}^N$ such that each counter is initialized with zero, $e_i$ denotes the number of red colored items of type $i$ while $n_i$ denotes the number of items of type $i$.
In addition, for every $i \in \br{1, \ldots, N}$, we define $\alpha_i$ to be an approximate fraction of the red items of type $i$ with respect to $n_i$, that is $0\leq \alpha_i \leq 1$ for all $i$. The invariant $e_i \approx \alpha_i \cdot n_i $ will be maintained throughout the whole process of the algorithm (in the sense that $|e_i -\alpha_i \cdot n_i| =O(1))$.

In addition to using $e_j$, $n_j$, and $\alpha_j$ for $j=1,2,\ldots,N$, there are auxiliary values calculated based on item types. The maximum number items of type $i$ that can be packed in one bin will be based on a parameter $\beta_i$ for every $i \in \{1 , \ldots, N\}$. This parameter will be used for blue items, since for them the maximum number will be packed (except for at most one bin for every type).
The amount of unused (free) space in bins filled with $\beta_i^d$ items from interval $I_i$ will be based on a value denoted by $\delta_i$ (this definition of $\delta_i$ will be slightly modified later). This value is defined according to the maximum size of any item of type $i$, which is $t_i$. This algorithm exploits this free space to pack red items of other types. Thus, $\delta_i=1- \beta_i \cdot t_i$. Note that this is the space in one dimension, while, for example, for $d=2$ the space is an $L$-shaped area whose width is $\delta_i$, see Figure~\ref{Param}.
We sometimes decide not to use the entire space of $\delta_i$ for red items. For simplicity of the algorithm and its analysis, we define the set
 $D = \left\lbrace \Delta_0=0,\Delta_1, \ldots, \Delta_k\right\rbrace$ to describe the set of spaces into which red items can be placed, such that $\Delta_k < 1/2$, $\Delta_i \leq \Delta_{i+1}$, and $\Delta_1 > 0$ for every $i$.
The set may contain all values of the form $\delta_i$ or just some of them.
Let $\phi : \br{1, \ldots, N} \to \br{0, \ldots, k}$ denote a mapping function from item types to their corresponding index $\Delta_j$, and for any $i \in \br{1, \ldots, N}$ denote by
$\Delta_{\phi(i)}$, the amount of space used to hold red items in a bin which holds blue items of type
$i$. We require that the function $\phi$ satisfies $\Delta_{\phi(i)} \leq \delta_i$. If $\phi(i) = 0$ holds, then no red items are accepted in bins filled with $\beta_i$ items.
For example, if $\Delta_1=0.28$, $\Delta_2=0.3$, $\Delta_3=0.32$ and $\delta_i=0.31$, we can choose $\phi(i)=2$. We could also choose $\phi(i)=1$, but the largest $j$ is chosen such that  $\Delta_{j} \leq \delta_i$, such that the space is used in the best way.
To ensure that for every red item there may potentially exist a bin to pack it, we require that $\alpha_i = 0$ for every $i \in \br{1, \ldots, N}$ such that $t_i > \Delta_k$.

We follow some of the literature of this type of algorithms, and use $\gamma_i$ to denote the maximum number of red items of type $i$ (for every $i \in \br{1 \ldots, N}$) that can be packed in the bin, where $\gamma_i = 0$ if $t_i > \Delta_k$ and $\gamma_i = \max\{1,\left\lbrace\Delta_1/t_i\right\rbrace\}$ otherwise. This value is the number of items that can fit in one dimension. For example, if $t_i=\frac 1{30}$ and $\Delta_1=0.21$, then $\gamma_i=6$, but in the case $t_i=0.22$, and if $\Delta_k =0.3$, we will have $\gamma_i=1$, which means that there will be just one red item of type $i$ next to blue items in each dimension.

To generalize the usage of $\gamma_i$ to $d$-dimensional bin packing, we define $\theta_i$ which denotes the maximum number of red items of type $i$ (for every $i \in \br{1, \ldots, N}$) that can be packed in a single $d$-dimensional bin as follows:
$$\theta_i =  \beta_i^d - (\beta_i - \gamma_i)^d \ .$$
For example, if $d=2$, $\beta_i=5$ and $\gamma_i=2$, we get $\theta_i=16$. This means that a cube with $\beta_i$ items of type $i$ packed in each dimension is created, and a smaller cube with $\beta_i-\gamma_i$ items in each dimension is removed to make space for other items. See Figure \ref{Param} for an illustration.
The definition $\theta_i=0$ for $t_i > \Delta_k$  means that  there is no place at any bin for red items of type $i$ since the blue items are too large. As mentioned in the preceding paragraph, we require that $\alpha_i = 0$ in these cases, i.e., all the items from interval $I_i$ are colored blue and there are no red items from interval $I_i$. It is possible that other values of $\alpha_i$ will also be equal to zero.

For simplicity, we redefine the values $\delta_i$ to be exactly the $\Delta_{\phi(i)}$ values (by possibly reducing some of these values). Thus, a red item of type $j$ can be packed with blue items of type $i$ if and only if $t_j  \leq \delta_i $.

The main ingredients for our algorithm are as follows.
\begin{enumerate}
\item A pair of integers $N$ and $k$, such that $N$ denotes number of intervals, and $k$ denotes the number of different sizes of spaces for red items.
\item Rational numbers $t_1=1>t_2>\cdots>t_N>t_{N+1}=0$, which denote the intervals boundaries, i.e., the $i$th interval is  $(t_{i+1},t_i]$.
\item Rational numbers $\alpha_1,\ldots,\alpha_N,\in [0,1]$, where for every $i \in \br{1,\ldots,N}$, $\alpha_i$ denotes the fraction of red items from the whole set of items in the $i$th interval.
\item Parameters $0 < \Delta_1 < \Delta_2 < \cdots < \Delta_k < 1/2$, which denote set of spaces into which red items can be placed.
\item  A function $\phi : \{1,\ldots,N\} \to \{0.\ldots,k\}$, which denotes a mapping function from item types to their corresponding indexes of spaces for red items. It always holds that $\Delta_{\phi(i)} \geq 1-\beta_i \cdot t_i$. For simplicity we denote $\Delta_{\phi(i)}$ by $\delta_i$.
\end{enumerate}

An item $x$ of size $s(x)$ has a type $\tau(x)$ where
$$\tau(x)=j \quad \Leftrightarrow \quad s(x) \in I_j.$$

\paragraph{Bin types.} The next table contains a description of the four types of bins used in the algorithm.
\begin{table}[H]
\centering
\begin{tabular}{|p{1cm} |p{5cm}||p{8cm} | }
\hline
 \multicolumn{3}{|c|}{Bin types} \\
 \hline
& Bin type & Description \\
 \hline
1.&$\{(i)|\phi(i)=0\}$  & Such bins include at most $\beta_i^d$ items, where all of them are considered blue and have type $i$. In the algorithm, they could also be defined as type $(i,?)$, but since they cannot receive red items of another type, the question mark would never be replaced.   \\
\hline
2.&$\{(i,?)|\phi(i)\neq 0\}$& Such bins include only blue items from interval $i$, where there is space left in it that can fit only red items of type $j$ that satisfy $t_j \leq  \delta_i$.  \\
 \hline
3.&$\{(?,j)|\alpha_j\neq 0\}$ & Such bins include only red items type of $j$, where there is space left in it that can fit only blue items of type $i$ that satisfies $\delta_i \geq  t_j $. \\
 \hline
4.&$\{(i,j)|\phi(i)\neq 0, \alpha_j\neq 0 \}$   & Such bins have both red items of type $j$ and blue items of  type $i$.\\
 \hline

\end{tabular}
\end{table}

Note that not all bin types have the required number of items. Bins that have a smaller number of items (less than $\beta_i^d$ for blue items of type $i$, or less than $\theta_j$ for red items of type $j$) is called indeterminate.

\begin{figure*}[!htb]
\begin{subfigure}[t]{.49\textwidth}
\centering
\captionsetup{width=.8\linewidth}
\includegraphics[width=0.8\textwidth]{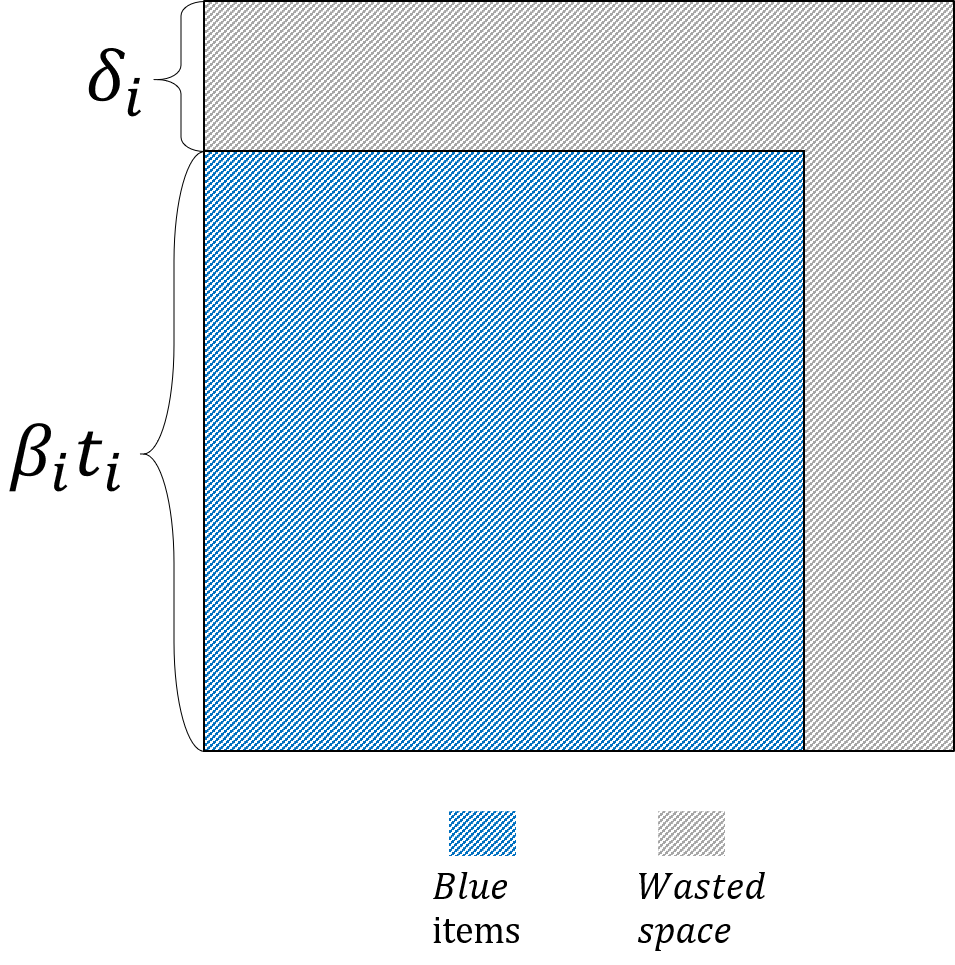}
\caption{An illustration of a type $(i)$ bin for $d=2$. In this case it holds that $\phi(i)=0$. The area that could be reserved for red items is $1-(\beta_i\cdot t_i)^2$. Since $\delta_i=0$, the space for red items is not used at all.}
\label{Param1}
\end{subfigure}
\begin{subfigure}[t]{.49\textwidth}
\centering
\captionsetup{width=.8\linewidth}
\includegraphics[width=0.8\textwidth]{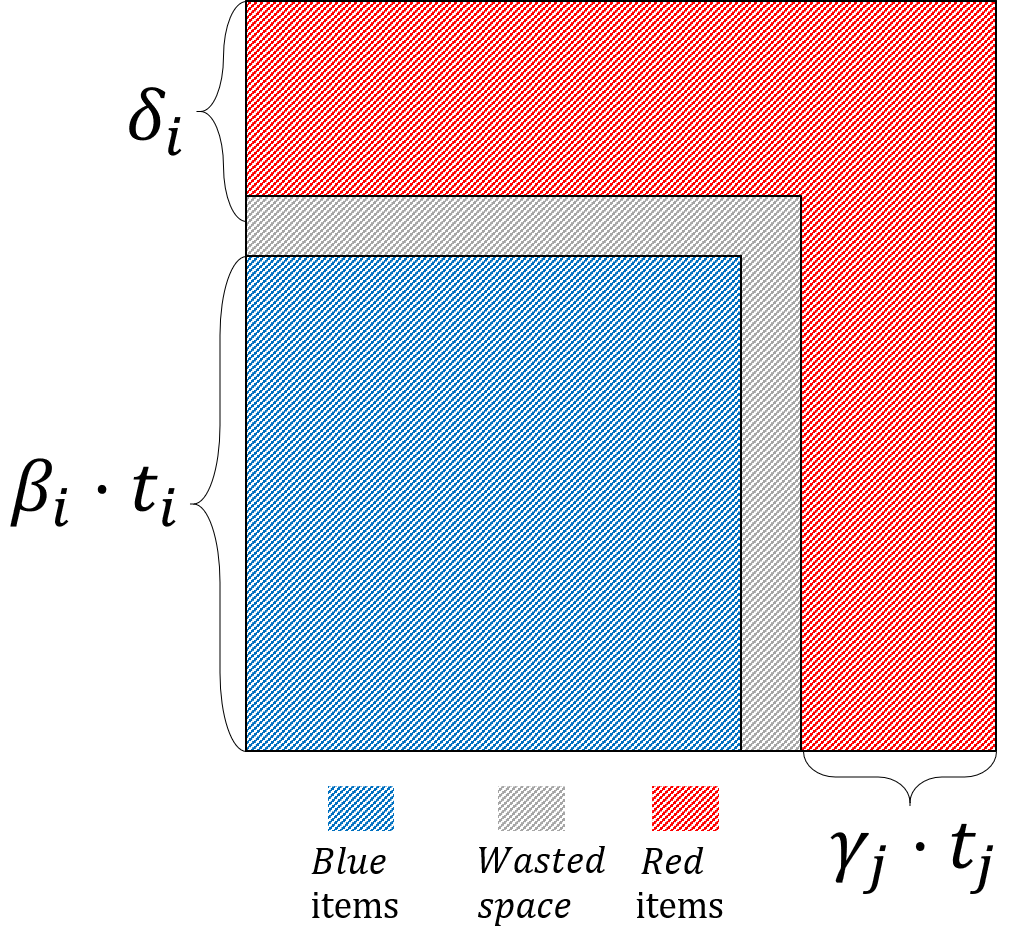}
\caption{An illustration of a type $(i,j)$ bin for $d=2$. The total area that could be reserved for red items is $1-(\beta_i\cdot t_i)^2$. The actual area reserved for red items is $1-(1-\delta_i)^2$, and the used space may be even smaller.
}
\label{Param2}
\end{subfigure}
\caption{Illustrations for bin types.}
\label{Param}
\end{figure*}




\paragraph{An overview of the code of Algorithm Extended Harmonic.}
In what follows, we give an overview of this algorithm, which is our main algorithm. The algorithm is defined for any dimension $d\geq 2$, and we will use this algorithm for the cases $d=2,3$.
We present the pseudo-code for our algorithm, see Algorithm \ref{EH}.
The algorithm colors each incoming item as blue or red. The coloring is based on the number of items of the same type that already arrived, such that the percentage of red items will be correct.  Specifically, for every type $i$, $n_i$ will be the total number of items of this type at each time, and $e_i$ will be the number of items of type $i$ whose color is red. Recall that the algorithm maintains the property $e_i = \alpha_i \cdot n_i$ approximately (since the numbers of items $e_i$ and $n_i$ are integers, while $\alpha_i \cdot n_i$ is not necessarily an integer). This is done by testing the ratio between $e_i$ and the new value of $n_i$ after an item of type $i$ arrives.

Types of bins are marked by pairs of indexes of types, where the first one is the type of blue items for this bin, and the second one is the type of red items. A type that was not decided yet appears as a question mark. Thus, a bin of type $(?,j)$ is a bin that already has at least one red item of type $j$ and no blue items.
For an item of type $i$ whose color is red, the algorithm checks whether there exists
an already existing bin that can be used to accommodate the new item. This has to be a bin that requires at least one additional item of type $i$ that is red. This may be a bin of type $(?,i)$ or $(j,i)$ for some type $j\neq i$, where its pre-determined number of items of type $i$ was not packed yet. We see such a bin as {\it open}.
If there is no such bin, it will check whether there is a bin with blue items but no red items, where red items of type $i$ can be accepted, and if indeed this is possible, such a bin is selected.
If there is no open bin to pack the new item (as a red item), then the algorithm  opens a new bin of type $(?,i)$  and packs the new item into it. Thus, the algorithm will pack the new item in the first open bin from the following ordered list of bins:

\begin{enumerate}
\item A bin of type $(?,i)$ or $(j,i)$ with less than $\theta_i$ red items in the bin.
\item A bin of type $(j,?)$ such that $\delta_j \geq \gamma_i \cdot t_i$.
\item A bin of $(?,i)$ (a new bin).
\end{enumerate}

The crucial part of this ordering that the algorithm avoids making new decisions as much as possible. The new item is packed into an open bin for red items of type $i$ if this is possible. If not, the algorithm still tries to use an existing bin, in order to avoid a situation where there are bins of types ($?,i)$ and $(j,?)$ which could be combined. Only if there is no other option, a new bin is introduced. In this case one can deduce that the current status of the output is such that all spaces that could receive red items of $i$ are already exhausted. Note that this can still change throughout the execution of the algorithm and we analyze only the final output.

For a new item of type $i$ whose color is blue, the case where this type cannot receive red items in its bins is easy. The item is either packed into a bin that does not have its full number of items, or if there is no such bin (which can also be called open), a new bin is opened. Such bins are denoted by type $(i)$. If this type can receive red items into the packing of its blue items, the algorithm checks whether there exists a bin that already received at least one blue item of type $i$, but it did not receive its full number of blue items, where this can be a bin of type $(i,?)$ or $(i,j)$ for some $j\neq i$. If there is no such bin, once again the algorithm prefers an existing bin with red items, and only if no such bin can accept blue items of type $i$, a new bin of type $(i,?)$ is opened. Thus, the algorithm will pack the new item into the first open bin from the following ordered list of bins:

\begin{enumerate}
\item A bin of type $(i,?)$ or $(i,j)$ with less than $\beta_i$ blue items in the bin.
\item A bin of type $(?,j)$ such that $\delta_j \geq \gamma_j \cdot t_j$.
\item A bin of type $(i,?)$ (a new bin).
    \end{enumerate}

Note that the algorithmic approach is almost identical to those of \cite{S1,HanYZ10}. The algorithm runs one copy of AssignSmall and packed every new small items with this algorithm.

\begin{algorithm}[h!]
\SetAlgoLined
\DontPrintSemicolon
\KwIn{List $L=\{\ell_1, \ell_2, \ldots, \ell_n\}$  of $d$-{\it dimensional} hyper-cube items such with rational sides in $(0,1]$.}
\KwOut{Set $D=\{d_1, d_2, \ldots, d_m\}$ of $d$-{\it dimensional} cube bins, such that these bins contain a valid packing of the items in $L$.}

$\forall z, e_z \leftarrow 0, n_z \leftarrow 0$\;
\For{$\ell\in L$}{
$i \leftarrow \tau(\ell), n_i = n_i+1$ \;
\If{$\ell$ is a small item}{pack $\ell$ using AssignSmall.}
\Else{
\uIf{$e_i < \left \lfloor \alpha_i \cdot n_i \right \rfloor$}
{
Color $\ell$ red;
$e_i = e_i +1$ \;
	\uIf{There exists an open bin of type $(?,i)$ or $(j,i)$ with fewer than $\theta_i$ red items in the bin}
{
		Pack $\ell$ into that bin.}
							
		\uElseIf
{There exists an open bin of type $(j,?)$ and $\delta_j\geq\gamma_i \cdot t_i$}
{
		Pack $\ell$ into that bin.\\
		Change the type of the bin to $(j,i)$.
}	
		\Else
{
		Open a new bin of type $(?,i)$ and pack $\ell$ there.
}
}}
\Else{
Color $\ell$ blue;\\
\If{$\phi(i)=0$}
{
\If{there exists an open bin of type $(i)$ with fewer than $\beta_i^d$ items}
{
        Pack $\ell$ into that bin.\\
}
\Else{
Open new bin of type $(i)$.\\
Pack $\ell$ into that bin.
}
}
\Else{
\If{There exists an open bin type $(i,?)$ or $(i,j)$ with fewer than $\beta_i^d$ blue items in the bin}
{
		Pack $\ell$ into that bin.
}}						
		\uElseIf
{There exists an open bin type $(?,j)$ and $\delta_i \geq \gamma_j\cdot t_j$}
{
		Pack $\ell$ into that bin.\\
		Change the bin type to $(i,j)$.
}	
		\Else
{
		Open a new bin of type $(i,?)$ and pack $\ell$ there.
}}}

\caption{Extended Harmonic.\label{EH}}
\end{algorithm}

\section{Weighting functions and results}
\label{44}
In what follows, we describe the weighting technique and present the specific weight functions which we use in our algorithm. As it was done in the past \cite{LeeLee85,S1,BBDEL_ESA18}, we split the different inputs into cases, based on a classification of the output.
We will define one weight function for every case, and the different weight functions are independent in the sense that every case will be analyzed separately. Obviously, all weight functions are based on the parameters of the algorithm, and those are common to all cases.

We assume here that $L$ consists of large items (only). Small items are packed separately, and the weight function used for them is not different from those used in the past. Specifically, the weight of a small item is $\frac{(M+1)^d}{(M^d-1)}$ times its area or volume.
In this section we only find the relation between the cost of the algorithm for large items and the total weight. Obviously, when we consider optimal solutions, and we find the relation of weights to their costs, we will consider small items as well, adding the weights of small items as well.

In the past \cite{S1,HanYZ10},  two weight functions were designed for the cases in the analysis with bin types that all of them have both red items and blue items. In the analysis, the two functions were compared in the sense that the better one was finally used. The intuition for the two functions was that either the cost of these bins is calculated as a part of the weights of the items that are blue, or it is taken into account in the weights of the items that are red. Informally, while these bins had both blue and red items, in this kind of analysis, the cost of the bins is either paid for by blue items or by red items.
The core of the technique which we use for our weight functions is the partitioning the cost of bins of type $(i,j)$ between red and blue items. This can be done in the cases described above, when there are no bins of types $(i,?)$ and $(?,j)$. For applying the method used here for the design of weight functions, we use a parameter $w$ ($0\leq w \leq 1$), where $w$ is the share of the blue items in bins where the cost is split, and $1-w$ is the share of red items. The value $w$ is not necessarily the same for all the cases, and it is typically different (any value can be used for any case and will lead to a correct proof, be we use values that allow us to prove upper bounds that are as tight as possible). The approach of previous work with two weight functions can be seen as the special case where the choice of $w$ had to be out of $\{0,1\}$, while we allow $w$ to be a rational number in $[0,1]$ and usually it is not an integer.

First, we define weighting functions for items such that the number of bins used by our algorithm is bounded by the total weight of the input sequence. Every weight function will be used for one case, where cases are defined later. For a weight function $U$, for any set $X$ of items, we let $U(X) = \sum_{p \in X} U(p)$.

The following lemma is similar to Lemma 2.2 of \cite{S1}.
\begin{lemma}\label{indeterminate}
The number of all indeterminate bins is $O(1)$, where the constant is independent of the input size.
\end{lemma}
\begin{proof}
The number of such bins is a linear function of the number of types. This holds due to steps $9$, $11$, $21$, $27$, $29$ of the algorithm.
\end{proof}

Given the last lemma, we assume that no such bins exist in the output.
Let $B_i$ and $R_i$ be the number of bins containing blue items and red items, respectively, for type $i$ (bins with both blue and red items are counted in two such values). Let $\lambda_i$ be the number of items of type $i$ in $L$. The algorithm keeps the proportion of red items out of all items for a given type almost exactly, up to a constant number of items for every type.  The next lemma was proved for the one-dimensional case \cite{S1}, and it holds for multiple dimensions since it deals with numbers of items, and not with sizes or possible ways to pack items. Since there are no red items for types $1,2,\ldots, 17$, we let $R_i=0$ for these types.

The next lemma is also similar to Lemma 2.2 of \cite{S1} (see also \cite{HanYZ10}).
\begin{lemma}\label{beta}
$B_i = \frac{1-\alpha_i}{\beta_i^d}\cdot\lambda_i + O(1),$ and $R_i = \frac{\alpha_i}{\theta_i}\cdot \lambda_i + O(1)$.
\end{lemma}

Let $Y$ denote number of  bins of type $(i,j)$ for all values of $i$ and $j$, i.e., the number of bins which have both red and blue items. The next property holds due to the double counting of such bins.

\begin{lemma}
\label{col}
$A(L) \leq \sum_{i} B_i + \sum_{i} R_i - Y.$
\end{lemma}

Let $q$ be the maximum index $i \leq 17$ such that there is at least one bin  at termination that satisfies the following condition: the bin is of the type $(i,?)$ if $i \notin \{2,3,\ldots,8\}$ and the bin is of the type
$(i,?)$ or $(20+i,?)$    for $2 \leq i \leq 8$. If there is no such $i$, we let $q=1$.
The motivation is to find whether there are bins with only blue items that are ready to receive red items. If there are such bins, we are interested in the largest value $\delta_g$ such that there is a bin of type $(g,?)$.
Let $e$ be the maximum index $j \geq 18$ such that there is at least one bin of the type $(?,j)$ at termination, and if there is no such $j$, we let $e=0$.  There will be no red items for type 18, and therefore in the case where $e>0$, where have $e \geq 19$.

\begin{lemma}\label{maxe}
If $2 \leq q\leq 9$, it holds that $e \leq 37-q$. If $10 \leq q\leq 16$, it holds that $e \leq 35-q$.
\end{lemma}
\begin{proof}
Assume that $2 \leq q\leq 9$, and consider the value $\delta_q$. The type $37-q+1$ has a right endpoint of $\delta_q$, and smaller items has smaller right endpoints for their intervals.
Thus, since there is a bin of type $(q,?)$ (or $(20+q,?)$ which is possible for $q\neq 9$), all red items that require space of $\delta_q$ or smaller are packed with blue items (otherwise, they could have been combined into a bin of type $(q,?)$ or $(20+q,?)$).

Assume that $10 \leq q\leq 16$, and consider the value $\delta_q$. The type $35-q+1$ has a right endpoint of $\delta_q$. Thus, since there is a bin of type $(q,?)$, all red items the require space of $\delta_i$ or smaller are packed with blue items.
\end{proof}

The next lemma holds by definition.

\begin{lemma}\label{easy}
Assume that $q \in \{2,3,\ldots,16\}$. For any $i\in \{q+1,\ldots,17\}$, there are no bins of type $(i,?)$, and for any  $j \geq e+1$ there are no bins of type $(?,j)$.
\end{lemma}

\begin{definition}
\label{f:w}
Let $0 \leq w \leq 1$ be a parameter used for the analysis, as explained above.
Let $q \in \{2,\ldots,16\}, e\in\{19,\ldots,151\}$. Define the weight of an item $p$ of size $x$ to be
\begin{equation*}
\label{last}
V_{e,q}(p) =
\begin{cases}
1, &\text{if $ x\in I_i,$ for $i = 1,\dots,q$}\\[1ex]
w, &\text{if $ x\in I_i,$ for $i =q+1,\ldots,17$}\\[1ex]
\frac{\alpha_i}{\theta_i} + \frac{1-\alpha_i}{\beta_i^d}, &\text{if $ x\in I_i,$ for $i =18,\ldots,e$}\\[1ex]
\frac{(1-w)\cdot \alpha_i}{\theta_i} + \frac{1-\alpha_i}{\beta_i^d}, &\text{if $ x\in I_i,$ for $i =e+1,\ldots,151$} \ .\\[1ex]
\end{cases}
\end{equation*}
\end{definition}
\begin{lemma}\label{theproof}
Let $q \in \{2,\ldots,16\}, e\in\{19,\ldots,151\},$ and let $V_{e,q}{(p)}$ be as in Definition~\ref{f:w} such that $e$ satisfies Lemma \ref{maxe} as its maximum value ($e=37-q$ if $q \leq 9$, and $e=35-q$ otherwise).
Then,  $A(L) \leq \sum_{p\in I_i} V_{e,q}(p) + O(1) .$
\end{lemma}

\begin{proof}

Consider the bin types for every $i$. For $1 \leq i \leq q$, there may be bins of types $(i,?)$ and $(i,j)$ for some values of $j$. For $q+1 \leq i \leq 17$, there may be only bins of types $(i,j)$ for some values of $j$.
There are no bins of type $(j,i)$ or $(?,i)$, since there are no red items for these types.
For $18 \leq i \leq e$, there may be bins of types $(i,?)$ and $(i,j)$ for some values of $j$, and bins of types $(?,i)$ and $(j,i)$ for some values of $j$.
For $i \geq e+1$, there may be bins of types $(i,?)$ and $(i,j)$ for some values of $j$, and bins of types $(j,i)$ for some values of $j$.

Since there are no red items of type $18$, the only bin types of the form $(?,j)$ that may exist are for $j$ such that $19 \leq j \leq e$, and the only bin types of the form $(i,?)$ that may exist are $\{1, 2, \ldots, q\}$. While there also may be bins of the type $(i,j)$ for $i \leq q$ or $j \leq e$, we cannot know for bins with blue items of a type $i \leq q$ if the bin also has red items, and we cannot know if a bin with red items of a type $j \leq e$ if it also has blue items. However, our analysis holds for all cases.

Let $X_1 = \sum_{i=q+1}^{17}B_i$ and $X_2 = \sum_{i=e+1}^{151}R_i$.
By the discussion above, all bins for these types contain both blue and red items, and therefore we get,
\begin{align}
\label{geqB}
Y \geq X_1  \mbox{ \ \ and \  \ } Y \geq X_2 .
\end{align}
This holds with inequality since there may be other bins with items of both colors.

We observe that for every $w \in [0,1]$ it holds that
\begin{align}
\label{geqBw}
(1-w)\cdot Y \geq (1-w)\cdot X_1 .
\end{align}
\begin{align}
\label{geqRw}
w\cdot Y \geq w\cdot X_2.
\end{align}
Hence, we get that
\begin{align}
\label{Yx}
Y \geq w\cdot X_2 + (1-w)\cdot X_1 .
\end{align}
Combining Lemma~\ref{col} with \eqref{Yx} we get that 
 \begin{align}
A(L)& \leq \sum_{i=1}^N B_i + \sum_{i=1}^N R_i - Y  \nonumber\\
& = \sum_{i=1}^{151} B_i + \sum_{i=18}^{151} R_i - (\sum_{i=q+1}^{17}B_i\cdot(1-w)+\sum_{i=e+1}^{151}R_i\cdot w).\label{eq1}
\end{align}
By that we get
\begin{align*}
A(L) &\leq \sum_{i \in \{1,\ldots,q,18,\ldots,151\}} B_i + \sum_{i=e+1}^{151} R_i \cdot (1-w) +  \sum_{i=q+1}^{17} B_i\cdot w + \sum_{i=19}^{e} R_i \\&
= \sum_{i \in \{1,\ldots, q,18\ldots 151\}} \frac{1-\alpha_i}{\beta_i^d} \cdot \lambda_i + \sum_{i=e+1}^{151} \frac{(1-w)\cdot \alpha_i}{\theta_i} \cdot \lambda_i \\ &  +\sum_{i \in {q+1,\ldots, 17}} \frac{(1-a_{i})\cdot w}{\beta_{i}^d} \cdot \lambda_i + \sum_{i \in {19,\ldots, e}} \frac{\alpha_i}{\theta_i}\cdot\lambda_i + O(1)\\&
=  \sum_{p\in I_i} V_{e,q}(p) + O(1),
\end{align*}

where the inequality holds by a simple rearrangement of~\eqref{eq1}, the first equality holds by Lemma~\ref{beta} and the definition of  $B_i$ and $R_i$, and the second equality holds by the definition of $V_{e,q}$.
\end{proof}
Next, we define weighting functions for large items such that
$$A(L) \leq \max_{1\leq i\leq17} W_i(L) + O(1) \ . $$
We split our proof into $17$ cases such that in each case we will use different weighting functions. Among $15$ of these cases, i.e.,  cases $2,3,\ldots,16$, we will define the weighting function using Definition~\ref{f:w} with respect to $e,q.$

\paragraph{Handling case $\boldsymbol{1}$:} This is the case where $q=1$. In this case it holds that all bins with blue items of sizes above $\frac 13$ that can be combined with red items were indeed combined with them.

In what follows, we define the weight of an item $p$ of size $x$ in this case.
\begin{equation*}
W_{1}(p) =
\begin{cases}
\frac{1-\alpha_i}{\beta_i^d}, &\text{if $ x\in I_i,$ for $i = 1,\ldots,18$}\\[1ex]
0, &\text{if $ x\in I_i,$ for $i =2,\ldots,17$}\\[1ex]
\frac{\alpha_i}{\theta_i} , &\text{if $ x\in I_i,$ for $i =22,\ldots,28$}\\[1ex]
\frac{\alpha_i}{\theta_i} + \frac{1-\alpha_i}{\beta_i^d}, &\text{if $ x\in I_i,$ for $i =19,\dots,21,29,\ldots,151$} \ .
\end{cases}
\end{equation*}

The definition of this case implies that bin types $(2,?), \ldots , (17,?)$ and $(22,?) ,\ldots, (28,?)$ do not exist. 
Hence, $Y \geq \sum_{i=1}^{17} B_i + \sum_{i=22}^{28} B_i$. We use the property $\alpha_i=0$ for $1 \leq i \leq 18$, and get

\begin{align*}
A(L)& \leq \sum_{i=1}^N B_i + \sum_{i=1}^N R_i - Y = \sum_{i=1}^{151} B_i + \sum_{i=1}^{151} R_i - \sum_{i = 2}^{17} B_i -\sum_{i=22}^{28} B_i  \\&
= \sum_{i=1, 18, 19, 20, 21, 29, \ldots, 151  } (B_i +R_i)+ \sum_{i=2,\ldots,17, 22\ldots, 28} R_i = \\ & = \sum_{i=1, 18, 19, 20, 21, 29, \ldots, 151  } (\frac{1-\alpha_i}{\beta_i^d}\cdot\lambda_i +\frac{\alpha_i}{\theta_i}\cdot \lambda_i)+ \sum_{i=2,\ldots,17, 22\ldots, 28}\frac{\alpha_i}{\theta_i}\cdot \lambda_i +O(1) \\ &  = \sum_{p\in I_i} W_{1}(p) + O(1) \ .
\end{align*}


\paragraph{Handling cases $\boldsymbol{2,3,\ldots,16}$:}In Table~\ref{Weighting}, we present each of the cases which rely on using both $e$ and $q$.
\begin{table}[H]
\centering
 \begin{tabular}{|c| c| c| c|}
 \hline
 Case $\#$ & the largest value of $e$ & $q$ & Weighting function \\ [0.5ex]
 \hline\hline
 $2$ & $35$ & $2$ & $W_{2}(L) =V_{35,2}(L)$  \\ \hline
 $3$ & $34$ & $3$ & $W_{3}(L) =V_{34,3}(L)$  \\ \hline
 $4$ & $33$ & $4$ & $W_{4}(L) =V_{33,4}(L)$  \\ \hline
 $5$ & $32$ & $5$ & $W_{5}(L) =V_{32,5}(L)$  \\ \hline
 $6$ & $31$ & $6$ & $W_{6}(L) =V_{31,6}(L)$  \\ \hline
 $7$ & $30$ & $7$ & $W_{7}(L) =V_{30,7}(L)$  \\ \hline
 $8$ & $29$ & $8$ & $W_{8}(L) =V_{29,8}(L)$  \\ \hline
 $9$ & $28$ & $9$ & $W_9(L) =V_{28,9}(L)$  \\ \hline
 $10$ & $25$ & $10$ & $W_{10}(L)=V_{25,10}(L)$  \\ \hline
 $11$ & $24$ & $11$ & $W_{11}(L) =V_{24,11}(L)$  \\ \hline
 $12$ & $23$ & $12$ & $W_{12}(L) =V_{23,12}(L)$  \\ \hline
 $13$ & $22$ & $13$ & $W_{13}(L) =V_{22,13}(L)$  \\ \hline
 $14$ & $21$ & $14$ & $W_{14}(L) =V_{21,14}(L)$  \\ \hline
 $15$ & $20$ & $15$ & $W_{15}(L) =V_{20,15}(L)$  \\ \hline
 $16$ & $19$ & $16$ & $W_{16}(L) = V_{19,16}(L)$ \\ \hline

 \end{tabular}
\caption{\label{Weighting} Weighting functions for cases $2,\ldots,16$.}
\end{table}
Since for every row in the table above, $t_e = 1 - t_{q+1}$ holds, substituting the values $e$, $q$ and $v_{e,q}$ of each row of the table above, into Lemma~\ref{theproof}, yields that $A(L) \leq \sum_{p\in I_i} V_{e,q}(p) + O(1).$

Note that in these weight functions we did not take into account the fact that the definition of $q$ considers also items of sizes in $(\frac 13,\frac 12]$ as blue items that can receive red items in their bins. The relevant cases are easy in the sense that the asymptotic competitive ratios for them are small even without reducing these weights (cases $2,\ldots, 7$), and reducing these weights will not change the competitive ratio of the algorithm.

\paragraph{Handling case $\boldsymbol{17}$:} In this case $q=17$. Any red item could have been combined into a blue bin of the form $(17,?)$, and thus, there are no $(?,j)$ bins at all. In what follows, we define the weight of an item $p$ of type  $i \leq 151$ in this case.
\begin{equation*}
W_{17}(p) =
\frac{1-\alpha_i}{\beta_i^d} \ .
\end{equation*}

Since the number of bins type $(?,j)$ is zero for any $j$, all the red items are packed in bins which include blue items. i.e., the only type of bins that may exist are $(i,j),(i),(i,?)$, which means that there are blue items packed into every bin. Hence, we get that $Y =\sum_{i} R_i$.
Which yields
\begin{align*}
A(L)& \leq \sum_{i} B_i + \sum_{i} R_i - Y = \sum_{i} B_i + \sum_{i} R_i - \sum_{i} R_i  \\&
= \sum_{i\in {1,\ldots,151}} B_i = \sum_{p\in I_i} W_{17}(p) + O(1),
\end{align*}
where first inequality holds by~\ref{col}, and the last equality holds by definition of $W_{17}$, $B_i$ and Lemma~\ref{beta}.
by the analysis above we get that

\begin{lemma}
\label{mainlemma}
$A(L) \leq \max_{1\leq i\leq17} W_i(L) + O(1).$
\end{lemma}

\subsection{Upper bounds on the asymptotic competitive ratio}\label{results}
In this section, we provide the $\alpha_i$ parameters for square and cube packing, respectively.
We also provide upper bounds on the asymptotic competitive ratio for each case in Table~\ref{final_results}.

For each $j \in \{1,2,\ldots,17\}$, we use the following integer program for obtaining an upper bound on the asymptotic competitive ratio,
\begin{eqnarray}
 & \nonumber \text{maximize} &    f_j(X)=\sum_{i=1}^{151} w_{i}\cdot x_{i} + \frac{112^d}{111^d-1}\bigg(1-\sum_{i=1}^{151} x_{i}\cdot t_{i+1}^d\bigg) \\
 & \nonumber \mbox{subject to} &   \\ & \ \ \ \ \ \ \ \ \ \ \ \ \ \ \ \ \ \ \ \ \ & \sum_{i=1}^{151} x_{i}\cdot t_{i+1}^d \leq 1  \label{areavolume} \\
 &   \ \ \ \ \ \ \ \ \ & \sum_{i=1}^{151} \left \lfloor{(t_{i+1}\cdot (u+1))}\right \rfloor^d \cdot x_{i} \leq u^d \text{ } \  \ \  \ \forall u \in \br{1, \ldots, 220} \label{220} \\ &  \nonumber \ \ \ \ \ \ \ \  \nonumber &  x_{i} \geq 0  \mbox{\ \ and \ } x_i \in \mathbb{Z}  \  \ \ \ \ \ \ \ \  \ \ \     \ \  \ \ \  \  \ \     \forall i \in \br{1, \ldots, 151}
\end{eqnarray}

Here X is a feasible set of items which fit into a single bin (of an optimal solution), $x_i$ is the number of items type $i$ in $X$, and $w_i$ is the weight of an item of type $i$, defined in the previous part of the section
by the function $W_i$.  The value $1-\sum_{i=1}^{111} x_{i}\cdot t_{i+1}^d$ is an upper bound on the total volume (or area) of all the small items in $X$, and by Lemma~\ref{leahsmall}, $\frac{112^d}{111^d-1}\cdot \bigg(1-\sum_{i=1}^{111} x_{i}\cdot t_{i+1}^d\bigg)$ is an upper bound of the total weight of all the small items in $X$.

The second type of constraints is based on a simple property that for an integer $u\geq 1$, no bin can contain more than $u^d$ items of size above $\frac 1{u+1}$ (see for example Claim 2.1 of \cite{Ep10}). For every item type, the constraint takes into account the number of independent items of size above $\frac 1{u+1}$ it can be split into. An item of type $i$ has a side above $\frac{1}{t+1}$, so every side can be split into $\lfloor \frac{t_{i+1}}{1/(u+1)} \rfloor$ parts. For example, an item of side above $\frac 12$ can be split into three items of sides above $\frac 16$ in every dimension.

\begin{table}[b!]
\begin{center}
\begin{tabular}{c || c c}
\hline
 & Square packing & Cube packing \\
\hline
case $1$ & $2.088447879968511$ & $2.5731896581108735$ \\
case $2$ & $1.9438375658626355$ & $2.45464218336544$ \\
case $3$ & $2.0109397168059324$ & $2.475823071455533$ \\
case $4$ & $1.9607242494316246$ & $2.455719344199358$ \\
case $5$ & $1.9942453743436321$ & $2.5115525001235937$ \\
case $6$ & $1.9875046382360564$ & $2.5339175799806912$ \\
case $7$ & $1.9554146240072456$ & $2.5016302664189443$ \\
case $8$ & $1.9441281429162531$ & $2.493821911539605$ \\
case $9$ & $2.0884478982863968$ & $2.5734762658161277$ \\
case $10$ & $2.0884277288254993$ & $2.5593413871191126$ \\
case $11$ & $2.088445077308426$ & $2.5567398601707696$ \\
case $12$ & $2.0876840226666538$ & $2.557631911023032$ \\
case $13$ & $2.0847781920964583$ & $2.5498950440578287$ \\
case $14$ & $2.07732977965866$ & $2.5226265870712448$ \\
case $15$ & $2.0656430335436333$ & $2.527717407098689$ \\
case $16$ & $2.0437751234561317$ & $2.5385458044738085$ \\
case $17$ & $2.088086287477056$ & $2.5718658072279847$ \\
\end{tabular}
\end{center}
\caption{Square and cube packing: upper bounds on the total weights for each case.}\label{final_results}
\end{table}

In order to obtain a slightly better result, we added two constraints of a different form to the integer program for the case $d=2$, as follows.

The first constraint is:
\begin{align}\label{con1}
\sum_{i=1}^{16}21\cdot x_i + \sum_{i=17}^{28}11\cdot x_i + \sum_{i=29}^{38}x_i \leq 57  \ .
\end{align}
The second constraint is:
\begin{align}\label{con2}
\sum_{i=1}^{16}80\cdot x_i + \sum_{i=17}^{28}30\cdot x_i + \sum_{i=29}^{37}10\cdot x_i + x_{38} \leq 190 \ .
\end{align}

\begin{lemma}
Conditions \eqref{con1} and \eqref{con2} hold for every valid bin of an optimal solution (for $d=2$).
\end{lemma}
\begin{proof}
We start with proving which type of contents of a bin each constraint excludes, given that the solution already satisfies the constraints of the original integer program.
From constraint \eqref{220} for $u=1,2,3,4$ we get \begin{equation} \sum_{i=1}^{17}  x_i \leq 1 \ , \label{111}  \end{equation}
\begin{equation}\sum_{i=1}^{8} 4 \cdot x_i + \sum_{i=9}^{28} x_i   \leq 4 \ , \label{222}  \end{equation}
\begin{equation}\sum_{i=1}^{17} 4\cdot x_i + \sum_{i=18}^{37} x_i \leq 9 \ , \label{333}  \end{equation}  \begin{equation} \mbox{and \ \ } \sum_{i=1}^{16} 9\cdot x_i + \sum_{i=17}^{18} 4 \cdot x_i + \sum_{i=19}^{38} x_i \leq 16 \ . \label{444}  \end{equation}

Let $G_{16}=\sum_{i=1}^{16} x_i$, $G_{28}=\sum_{i=17}^{28} x_i$, $G_{38}=\sum_{i=29}^{38} x_i$. Since all variables are non-negative and integral,
by \eqref{111}, we have  $G_{16}   \leq 1$, and therefore we have either  $ G_{16}= 1$ (one of the corresponding variables is equal to $1$) or $G_{16} = 0$ (all these variables are equal to zero). By \eqref{222}, we have $G_{16}+G_{28}  \leq 4$.  By \eqref{444}, we have $ 9\cdot G_{16}  + G_{28}+G_{38} \leq 16$. For proving \eqref{con1}, we will show that $21 \cdot G_{16}+ 11 \cdot G_{28}+ G_{38} \leq 57$ holds for all cases due to the already existing constraints, except for one case that we prove separately. Indeed, if $G_{16}=0$, by the last constraints (those that are based on \eqref{222} and \eqref{444}), we have $ G_{28}+G_{38} \leq 16$ and $G_{28} \leq 4$, and we get  $$21 \cdot G_{16}+ 11 \cdot G_{28}+ G_{38} = 11 \cdot G_{28}+ G_{38} = 10\cdot G_{28}+(G_{28}+G_{38}) \leq 10\cdot 4+ 16= 56 \ . $$

If $G_{16}=1$, we get $G_{28} \leq 3$ by \eqref{222}. If $G_{28}\leq 2$, we get (using \eqref{444}), $$21 \cdot G_{16}+ 11 \cdot G_{28}+ G_{38} = 12 \cdot G_{16}+ 10 \cdot G_{28} +  (9\cdot G_{16}+ G_{28}+ G_{38}) \leq 12 +20 + 16 = 48 \ . $$
If $G_{16}=1$ and $G_{28}=3$, the constraint $ 9\cdot G_{16}  + G_{28}+G_{38} \leq 16$ is equivalent to $G_{38} \leq 4$. If $G_{38}\leq 3$, we get  $$21 \cdot G_{16}+ 11 \cdot G_{28}+ G_{38} \leq 21+33+3 =57 \ . $$
Thus, to complete the proof of the constraint, it is required to prove that the remaining case $G_{16}=1$, $G_{28}=3$, and $G_{38}=4$ is impossible, since this is the only remaining case. This is done after the discussion of the second constraint.

For the second constraint \eqref{con2}, let $H_{38}=x_{38}$ and $H_{37}=G_{38}-H_{38}=\sum_{i=29}^{37} x_i$. We will also use the two new variables, whose sum is $G_{38}$. From \eqref{333} we have $4 \cdot G_{16}+ G_{28}+H_{37} \leq 9$, and by \eqref{444} we have $$9\cdot G_{16}+ G_{28}+ H_{37}+H_{38} \leq 16 \ . $$
We prove the constraint $80\cdot G_{16}+30 \cdot G_{28}+10\cdot H_{37}+H_{38} \leq 190$.

If $G_{16}=0$, by also using $G_{28}\leq 4$, we have $$80\cdot G_{16}+30 \cdot G_{28}+10\cdot H_{37}+H_{38} =20\cdot G_{28}+9\cdot(G_{28}+ H_{37})+(G_{28}+ H_{37}+H_{38})$$  $$\leq 80+ 81 +16 = 177 \ .$$
If $G_{16}=1$ and $G_{28} \leq 2$, we have $$80\cdot G_{16}+30 \cdot G_{28}+10\cdot H_{37}+H_{38} $$ $$ =35\cdot G_{16}+20\cdot G_{28}+9\cdot(4\cdot G_{16}+G_{28}+ H_{37})+(9\cdot G_{16}+G_{28}+ H_{37}+H_{38}) \leq 35+ 40+ 81 +16 = 172 \ . $$
In the remaining case, where $G_{16}=1$ and $G_{28}=3$ hold, we show later that $G_{38}\leq 3$, and $H_{37}\leq 2$. Moreover, we will show that if indeed  $G_{16}=1$, $G_{28}=3$, and $H_{37}=2$ hold, then $H_{38}=0$.
We would like to show that in the case $H_{37}\leq 1$, the constraint still holds, and it holds also if $H_{37}=2$ and $H_{38}=0$.
In the first case, $80\cdot G_{16}+30 \cdot G_{28}+10\cdot H_{37}+H_{38} = 80+90+ 9 \cdot H_{37}+ (H_{37}+H_{38}) \leq 182$.
In the second case, substitution of the exact values yields exactly $190$.

We now show that in the case $G_{16}=1$ and $G_{28}=3$, it holds that $G_{38}\leq 3$.
Assume that a bin contains an item of size in $(0.6,1]$, four items of sizes above $\frac 13$ (one of which is the item of size above $0.6$), and eight items of sizes above $\frac 15$ (out of which, four are larger than $\frac 13$).

All items have sizes above $0.2$. We start with claiming that the item of size above $0.6$ is packed in a corner of the bin. If this is not the case, it can be moved to a corner if there is no item blocking it. Since its side is above $0.6$, it cannot be the case that there is another item both below and above it, so it can be shifted in one direction until it reaches a side of the bin (the top or the bottom). Similarly, there cannot be both an item to its right and to its left, so it can be moved to the left or to the right side of the bin. By rotating the bin, assume that it is packed into the top left corner.

Draw two lines as follows: a horizontal line with distance $0.2$ from the bottom of the bin, and a vertical one, with distance $0.2$ from the right side of the bin. These lines do not intersect the item of size above $0.6$, but we claim that they intersect the interior of all other items. 
Every item that the two lines do not intersect must be contained in an $L$-shaped area whose height and width are below $0.2$, and the corner also allows the packing of an item whose side is smaller than $0.2$. Since there are no such items, all other items have an intersection with one of the lines or both.
In addition to the item of size above $0.6$, the bin has three other items of sides are above $\frac 13$. Thus, out of the two lines,  there is a line that intersects two such items (it cannot intersect all three, but it is possible that each one of the lines intersects two such items). Given the item sizes, a line that intersects two such items can intersect only one additional (smaller) item (this item has size in $(\frac 15,\frac 13]$). The third item of size above $\frac 13$ is intersected by the other line, and that line can intersect three other items in total. However, if it intersects four items in total, one of them is in fact intersected by both lines, since the total size of items not intersected by the other line is below $0.8$. Thus, in total, there are at most six intersected items, three of them have sizes above $\frac 13$, and at most three of them have sides in $(\frac 15,\frac 13]$.
This proves that in this case $G_{38}\leq 3$.

Now, we show that if there are two items with sides in $(\frac 14,\frac 13]$, there cannot be a third item. If there are at most five intersected items, we are done.
When there are six items of sides in $(\frac 15,\frac 12]$, four of them are intersected by one of the lines, such that only one of them has size in $(\frac 13,\frac 12]$. If there are two items with sides in $(\frac 14,\frac 13]$ and an item whose side is above $\frac 15$, the total is above $1$, which is impossible since the items can overlap only in the boundaries. This shows that in the case $G_{16}=1$ and $G_{28}=3$, it holds that $H_{37}\leq 2$, and if $H_{37} =  2$, it also holds that $H_{38}=0$.
\end{proof}

The next theorem states our main result.
\begin{theorem}
The asymptotic performance ratio of Algorithm Extended Harmonic for square packing is at most $2.0885$, while for cube packing is at most $2.5735$.
\end{theorem}
\begin{proof}
We set the parameters $\alpha_i$ according to Table~\ref{alpha} for square and cube packing. For each case we applied a simple integer program solver in order to find the worst case bound.
We obtain the results for square and cube packing, as described in Table \ref{final_results}.
Hence, we get that $A(L) \leq 2.5735 \cdot OPT(L) +O(1)$ for cube packing, and $A(L) \leq 2.0885 \cdot OPT(L) +O(1)$ for square packing.
\end{proof}
\section{The parameters for our algorithms}
\label{66}

In this section, we provide the interval partition and the  parameters $\alpha_i$ for square and cube packing, respectively. We also include values used in the algorithms that are based on the intervals.

\begin{longtable}{ |P{3cm}||P{3cm}|P{2cm}|P{1.5cm}|P{1.5cm}|P{1.5cm}|  }
\caption{Intervals and auxiliary values used in our algorithm.}\label{Intervals}\\
 \hline
 $i$& $(t_{i+1},t_ {i}]$	& $\delta_i$& $\beta_i$& $\gamma_i$&$\phi(i)$ \\
 \hline
$1$	& $(0.7,1]$	&$0$&$1$&$0$&$0$\\
\hline
$2$  &	$(0.6875,0.7]$ & $0.3$	&$1$&$0$&$1$\\
\hline
$3$  &$(0.675,0.6875]	$& $0.3125$&$1$&$0$&$2$\\
\hline
$4$	& $(0.67,0.675]$	& $0.325$&	$1$&$0$&$3$\\
\hline
$5$  &	$(0.668,0.67]$& $0.33$ &$1$&$0$&$4$\\
\hline
$6$  &	$(0.667,0.668]$	& $0.332$&$1$&$0$&$5$\\
\hline
$7$  &	$(0.6667,0.667]$	& $0.333$&$1$&$0$&$6$\\
\hline
$8$  &	$(\frac{2}{3},0.6667]$	& $0.3333$&$1$&$0$&$7$\\
\hline
$9$  &	$(0.666,\frac{2}{3}]$	& $\frac{1}{3}$&$1$&$0$&$8$\\
\hline
$10$  &	$(0.665,0.666]$ & $0.334$ &$1$&$0$&$9$\\
\hline
$11$  &	$(0.6625,0.665]$ & $0.335$	&$1$&$0$&$10$\\
\hline
$12$  &	$(0.65625,0.6625]$ & $0.3375$	&$1$&$0$&$11$\\
\hline
$13$  &	$(0.65,0.65625]$ & $0.34375$	&$1$&$0$&$12$\\
\hline
$14$  &	$(\frac{7}{11},0.65]$ & $0.35$ &$1$&$0$&$13$\\
\hline
$15$  &	$(0.625,\frac{7}{11}]$	& $\frac{4}{11}$ &$1$&$0$&$14$\\
\hline
$16$  &	$(0.6,0.625]$ & $0.375$&$1$&$0$&$15$\\
\hline
$17$  &	$(0.5,0.6]$ & $0.4$ &$1$&$0$&$16$\\
\hline
$18$  &	$(0.4,0.5]$ & $0$ &$2$&$0$&$0$\\
\hline
$19$  &	$(0.375,0.4]$ & $0$ &$2$&$1$&$0$\\
\hline
$20$  &	$(\frac{4}{11},0.375]$ & $0$ &$2$&$1$&$0$\\
\hline
$21$  &	$(0.35,\frac{4}{11}]$	& $0$ &$2$&$1$&$0$\\
\hline
$22$  &	$(0.34375,0.35]$ & $0.3$ &$2$&$1$&$1$\\
\hline
$23$  &	$(0.3375,0.34375]$ & $0.3125$ &$2$&$1$&$2$\\
\hline
$24$  &	$(0.335,0.3375]$ & $0.325$ &$2$&$1$&$3$\\
\hline
$25$  &	$(0.334,0.335]$ & $0.33$ &$2$&$1$&$4$\\
\hline
$26$  &	$(0.3335,0.334]$ & $0.332$ &$2$&$1$&$5$\\
\hline
$27$  &	$(0.33335,0.3335]$ & $0.333$ &$2$&$1$&$6$\\
\hline
$28$  &	$(\frac{1}{3},0.33335]$ & $0.3333$ &$2$&$1$&$7$\\
\hline
$29$  &	$(0.3333,\frac{1}{3}]$	& $0$ &$3$&$1$&$0$\\
\hline
$30$  &	$(0.333,0.3333]$ & $0$ &$3$&$1$&$0$\\
 \hline
 $i$& $(t_{i+1},t_ {i}]$	& $\delta_i$& $\beta_i$& $\gamma_i$&$\phi(i)$ \\
\hline
$31$  &	$(0.332,0.333]$ & $0$ &$3$&$1$&$0$\\
\hline
$32$  &	$(0.33,0.332]$ & $0$ &$3$&$1$&$0$\\
\hline
$33$  &	$(0.325,0.33]$ & $0$ &$3$&$1$&$0$\\
\hline
$34$  &	$(0.3125,0.325]$ & $0$	&$3$&$1$&$0$\\
\hline
$35$  &	$(0.3,0.3125]$ & $0$ &$3$&$1$&$0$\\
\hline
$36$  &	$(\frac{3}{11},0.3]$ & $0$	&$3$&$1$&$0$\\
\hline
$37$  &	$(\frac{1}{4},\frac{3}{11}]$	& $0$&$3$&$1$&$0$\\
\hline
$38$  &	$(\frac{1}{5},\frac{1}{4}]$	& $0$&$4$&$1$&$0$\\
\hline
$39$  &	$(\frac{2}{11},\frac{1}{5}]$	& $0$&$5$&$1$&$0$\\
\hline
$40$  &	$(\frac{1}{6},\frac{2}{11}]$	& $0$&$5$&$1$&$0$\\
\hline
$41$  & $(0.15,\frac{1}{6}]$	& $0$ & $6$ &$1$&$0$\\
\hline
$42$  & $(\frac{1}{7},0.15]$	 & $0$ &$6$&$2$&$0$\\
\hline
$43$ & $(\frac{1}{8},\frac{1}{7}]$ & $0$ & $7$ & $2$&$0$ \\
\hline
$44$ & $(\frac{1}{9},\frac{1}{8}]$ & $0$ & $8$ & $2$&$0$ \\
\hline
$45$ & $(\frac{1}{10},\frac{1}{9}]$ & $0$ & $9$ & $2$&$0$ \\
\hline
$46$ & $(\frac{1}{11},\frac{1}{10}]$ & $0$ & $10$ & $3$ &$0$\\
\hline
$47$ & $(\frac{1}{12},\frac{1}{11}]$ & $0$ & $11$ & $3$&$0$ \\
\hline
$48$ & $(\frac{1}{13},\frac{1}{12}]$ & $0$ & $12$ & $3$ &$0$\\
\hline
$49$ & $(0.075,\frac{1}{13}]$ & $0$ & $13$ & $3$&$0$ \\
\hline
$50$ & $(\frac{1}{14},0.075]$ & $0$ & $13$ & $4$ &$0$\\
\hline
$51$ & $(\frac{1}{15},\frac{1}{14}]$ & $0$ & $14$ & $4$ &$0$\\
\hline
$52$ & $(\frac{1}{16},\frac{1}{15}]$ & $0$ & $15$ & $4$&$0$ \\
\hline
$53$ & $(0.06,\frac{1}{16}]$ & $0$ & $16$ & $4$&$0$ \\
\hline
$54$ & $(\frac{1}{17},0.06]$ & $0$ & $16$ & $5$&$0$ \\
\hline
$55,\ldots,61$ & $(\frac{1}{i-37},\frac{1}{i-38}]$ & $0$ &  Table~\ref{betaa}& Table~\ref{betaa}&$0$\\
\hline
$62$ & $(\frac{1}{24},\frac{3}{70}]$ & $0$ & $23$ & $7$&$0$ \\
\hline
$63$ & $(\frac{1}{25},\frac{1}{24}]$ & $0$ & $24$ & $7$ &$0$\\
\hline
$64$ & $(\frac{1}{26},\frac{1}{25}]$ & $0$ & $25$ & $7$ &$0$\\
\hline
$65$ & $(\frac{3}{80},\frac{1}{26}]$ & $0$ & $26$ & $7$&$0$ \\
\hline
$66$ & $(\frac{1}{27},\frac{3}{80}]$ & $0$ & $26$ & $8$&$0$ \\
\hline
$67,\ldots,73$ & $(\frac{1}{i-39},\frac{1}{i-40}]$ & $0$ &  Table~\ref{betaa}& Table~\ref{betaa}&$0$\\
\hline
$74$ & ($\frac{1}{34}$,$0.03$] & $0$ & $33$ & $10$&$0$ \\
\hline
$75,\ldots,151$ &$(\frac{1}{i-40},\frac{1}{i-41}]$ & $0$ &  Table~\ref{betaa}& Table~\ref{betaa}&$0$\\
\hline
\end{longtable}


\begin{longtable}{|P{1cm}|P{1cm}|P{1cm}||P{1cm}|P{1cm}|P{1cm}||P{1cm}|P{1cm}|P{1cm}|}
\caption{The values $\beta_i$ and $\gamma_i$  for Table~\ref{Intervals}.}\label{betaa}\\
\hline
 $\boldsymbol{i}$& $\beta_i$& $\gamma_i$&   $\boldsymbol{i}$& $\beta_i$& $\gamma_i$ &   $\boldsymbol{i}$& $\beta_i$& $\gamma_i$ \\
\hline
$\mathbf{55}$ & $17$ & $5$ & $\mathbf{92}$ & $51$ & $15$ & $\mathbf{122}$ & $81$ & $24$  \\
\hline
$\mathbf{56}$ & $18$ & $5$ & $\mathbf{93}$ & $52$ & $15$ & $\mathbf{123}$ & $82$ & $24$  \\
\hline
$\mathbf{57}$ & $19$ & $5$ & $\mathbf{94}$ & $53$ & $15$ & $\mathbf{124}$ & $83$ & $24$  \\
\hline
$\mathbf{58}$ & $20$ & $6$ & $\mathbf{95}$ & $54$ & $16$ & $\mathbf{125}$ & $84$ & $25$  \\
\hline
$\mathbf{59}$ & $21$ & $6$ & $\mathbf{96}$ & $55$ & $16$ & $\mathbf{126}$ & $85$ & $25$  \\
\hline
$\mathbf{60}$ & $22$ & $6$ & $\mathbf{97}$ & $56$ & $16$ & $\mathbf{127}$ & $86$ & $25$  \\
\hline
$\mathbf{61}$ & $23$ & $6$ & $\mathbf{98}$ & $57$ & $17$ & $\mathbf{128}$ & $87$ & $26$  \\
\hline
$\mathbf{67}$ & $27$ & $8$ & $\mathbf{99}$ & $58$ & $17$ & $\mathbf{129}$ & $88$ & $26$  \\
\hline
$\mathbf{68}$ & $28$ & $8$ & $\mathbf{100}$ & $59$ & $17$ & $\mathbf{130}$ & $89$ & $26$  \\
\hline
$\mathbf{69}$ & $29$ & $8$ & $\mathbf{101}$ & $60$ & $18$ & $\mathbf{131}$ & $90$ & $27$  \\
\hline
$\mathbf{70}$ & $30$ & $9$ & $\mathbf{102}$ & $61$ & $18$ & $\mathbf{132}$ & $91$ & $27$  \\
\hline
$\mathbf{71}$ & $31$ & $9$ & $\mathbf{103}$ & $62$ & $18$ & $\mathbf{133}$ & $92$ & $27$  \\
\hline
$\mathbf{72}$ & $32$ & $9$ & $\mathbf{104}$ & $63$ & $18$ & $\mathbf{134}$ & $92$ & $27$  \\
\hline
$\mathbf{73}$ & $33$ & $9$ & $\mathbf{105}$ & $64$ & $19$ & $\mathbf{135}$ & $94$ & $28$  \\
\hline
$\mathbf{75}$ & $34$ & $10$ & $\mathbf{106}$ & $65$ & $19$ & $\mathbf{136}$ & $95$ & $28$  \\
\hline
$\mathbf{76}$ & $35$ & $10$ & $\mathbf{107}$ & $66$ & $19$ & $\mathbf{137}$ & $96$ & $28$  \\
\hline
$\mathbf{77}$ & $36$ & $10$ & $\mathbf{108}$ & $67$ & $20$ & $\mathbf{138}$ & $97$ & $29$  \\
\hline
$\mathbf{78}$ & $37$ & $11$ & $\mathbf{109}$ & $68$ & $20$ & $\mathbf{139}$ & $98$ & $29$  \\
\hline
$\mathbf{79}$ & $38$ & $11$ & $\mathbf{110}$ & $69$ & $20$ & $\mathbf{140}$ & $98$ & $29$  \\
\hline
$\mathbf{80}$ & $39$ & $11$ & $\mathbf{111}$ & $70$ & $21$ & $\mathbf{141}$ & $100$ & $30$  \\
\hline
$\mathbf{81}$ & $40$ & $12$ & $\mathbf{112}$ & $71$ & $21$ & $\mathbf{142}$ & $101$ & $30$  \\
\hline
$\mathbf{82}$ & $41$ & $12$ & $\mathbf{113}$ & $72$ & $21$ & $\mathbf{143}$ & $102$ & $30$  \\
\hline
$\mathbf{83}$ & $42$ & $12$ & $\mathbf{114}$ & $73$ & $21$ & $\mathbf{144}$ & $103$ & $30$  \\
\hline
$\mathbf{84}$ & $43$ & $12$ & $\mathbf{115}$ & $74$ & $22$ & $\mathbf{145}$ & $104$ & $31$  \\
\hline
$\mathbf{85}$ & $44$ & $13$ & $\mathbf{116}$ & $75$ & $22$ & $\mathbf{146}$ & $105$ & $31$  \\
\hline
$\mathbf{86}$ & $45$ & $13$ & $\mathbf{117}$ & $76$ & $22$ & $\mathbf{147}$ & $106$ & $31$  \\
\hline
$\mathbf{87}$ & $46$ & $13$ & $\mathbf{118}$ & $77$ & $23$ & $\mathbf{148}$ & $107$ & $32$  \\
\hline
$\mathbf{88}$ & $47$ & $14$ & $\mathbf{119}$ & $78$ & $23$ & $\mathbf{149}$ & $108$ & $32$  \\
\hline
$\mathbf{89}$ & $48$ & $14$ & $\mathbf{120}$ & $79$ & $23$ & $\mathbf{150}$ & $109$ & $32$  \\
\hline
$\mathbf{90}$ & $49$ & $14$ & $\mathbf{121}$ & $80$ & $24$ & $\mathbf{151}$ & $110$ & $33$  \\
\hline
$\mathbf{91}$ & $50$ & $15$ & & & & & &  \\
\hline
\end{longtable}
\begin{longtable}{ |P{3cm}||P{5cm} ||P{5cm}| }
\caption{Values of $\alpha_i$ that are parameters of our algorithms for square and cube packing.}\label{alpha}\\
\hline
 & Square packing & Cube packing \\
\hline
$\alpha_{i}$ for $1\leq i \leq 18$ & $0$ & $0$ \\
\hline
$\alpha_{19}$ & $0.11526431542309074$ & $0.23560671174940934$ \\
\hline
$\alpha_{20}$ & $0.17175402209391144$ & $0.24349456708719025$ \\
\hline
$\alpha_{21}$ & $0.14364948238440467$ & $0.011054757786850555$ \\
\hline
$\alpha_{22}$ & $0.17775964679070577$ & $0.09233137770530553$ \\
\hline
$\alpha_{23}$ & $0.16247599807416024$ & $0.10296544873687286$ \\
\hline
$\alpha_{24}$ & $0.17013150154133094$ & $0.09980866333707894$ \\
\hline
$\alpha_{25}$ & $0.17218382694021506$ & $0.11275956304754697$ \\
\hline
$\alpha_{26}$ & $0.17186065470253054$ & $0.10573246664180191$ \\
\hline
$\alpha_{27}$ & $0.1712411485735466$ & $0.21831169314212995$ \\
\hline
$\alpha_{28}$ & $0.17115325420709004$ & $0.16810602509149197$ \\
\hline
$\alpha_{29}$ & $0.011808683266528508$ & $0.28469363087983357$ \\
\hline
$\alpha_{30}$ & $0.08864616236688028$ & $0.46134537517964436$ \\
\hline
$\alpha_{31}$ & $0.0746578085809842$ & $0.4754821887062161$ \\
\hline
$\alpha_{32}$ & $0.1392973955221088$ & $0.4834778208599464$ \\
\hline
$\alpha_{33}$ & $0.20463684875950888$ & $0.38230203454521344$ \\
\hline
$\alpha_{34}$ & $0.11988863237025116$ & $0.20815458494242878$ \\
\hline
$\alpha_{35}$ & $0.1489855469399089$ & $0.2094357013281899$ \\
\hline
$\alpha_{36}$ & $0.42658319200096906$ & $0.6476643335428202$ \\
\hline
$\alpha_{37}$ & $0.3313855159770591$ & $0.4846417112019235$ \\
\hline
$\alpha_{38}$ & $0.26591984078589526$ & $0.3459551479018446$ \\
\hline
$\alpha_{39}$ & $0.23652286713889142$ & $0.1967822914561262$ \\
\hline
$\alpha_{40}$ & $0.17320945474790095$ & $0.22903844377204607$ \\
\hline
$\alpha_{41}$ & $0.2907287245318693$ & $0.38585033090166515$ \\
\hline
$\alpha_{42}$ & $0.27690915366279856$ & $0.2633509344925706$ \\
\hline
$\alpha_{43}$ & $0.35186597263941155$ & $0.37148866892244403$ \\
\hline
$\alpha_{44}$ & $0.28487022531216166$ & $0.3228819685751433$ \\
\hline
$\alpha_{45}$ & $0.3405383352070134$ & $0.294966161863426$ \\
\hline
$\alpha_{46}$ & $0.13927977565087557$ & $0.11613078486074929$ \\
\hline
$\alpha_{47}$ & $0.12478043051170912$ & $0.21976007519116803$ \\
\hline
$\alpha_{48}$ & $0.17368906765817593$ & $0.2367222519372697$ \\
\hline
$\alpha_{49}$ & $0.049341692986982266$ & $0.06874946889000572$ \\
\hline
$\alpha_{50}$ & $0.21756972846743544$ & $0.30801878565803864$ \\
\hline
$\alpha_{51}$ & $0.15176378068862706$ & $0.10874307802527139$ \\
\hline
$\alpha_{52}$ & $0.27986004047748236$ & $0.34382124885682674$ \\
\hline
$\alpha_{53}$ & $0.09140290314421057$ & $0.19822255924214388$ \\
\hline
$\alpha_{54}$ & $0.16115290643799296$ & $0.21657253679087018$ \\
\hline
$\alpha_{55}$ & $0.10509477906408826$ & $0.21064008575188697$ \\
\hline
$\alpha_{56}$ & $0.07908677596102542$ & $0.5286073975827003$ \\
\hline
$\alpha_{57}$ & $0.06049271754448721$ & $0.23593465027098925$ \\
\hline
$\alpha_{58}$ & $0.027902842302122366$ & $0.10627837309910759$ \\
\hline
$\alpha_{59}$ & $0.03757222734769261$ & $0.08778737037136902$ \\
\hline
$\alpha_{60}$ & $0.044034294107809235$ & $0.0628782883568702$ \\
\hline
$\alpha_{61}$ & $0.04169873464584284$ & $0.07892306409577904$ \\
\hline
$\alpha_{62}$ & $0.045855398808323844$ & $0.06811428634665145$ \\
\hline
$\alpha_{63}$ & $0.03268220721227799$ & $0.08934119933293255$ \\
\hline
$\alpha_{64}$ & $0.020287554239005412$ & $0.10985985543445637$ \\
\hline
$\alpha_{65}$ & $0.03662245261759983$ & $0.16657268323184893$ \\
\hline
$\alpha_{66}$ & $0.05299014948250891$ & $0.16370099694324725$ \\
\hline
$\alpha_{67}$ & $0.05837546569384355$ & $0.14763245122124014$ \\
\hline
$\alpha_{68}$ & $0.06021197613253543$ & $0.1671268810238925$ \\
\hline
$\alpha_{69}$ & $0.05286287383333055$ & $0.18510082544610912$ \\
\hline
$\alpha_{70}$ & $0.041141831190207534$ & $0.011723129997064097$ \\
\hline
$\alpha_{71}$ & $0.025858702537442546$ & $0.02425242847273701$ \\
\hline
$\alpha_{72}$ & $0.03667621572334345$ & $0.011268687510284647$ \\
\hline
$\alpha_{73}$ & $0.05790545597682889$ & $0.01566133856254459$ \\
\hline
$\alpha_{74}$ & $0.0249935407107143$ & $0.0023807218784999695$ \\
\hline
$\alpha_{75}$ & $0.05090633446809589$ & $0$ \\
\hline
$\alpha_{76}$ & $0.04180489086300371$ & $0.014065837749926702$ \\
\hline
$\alpha_{77}$ & $0.0598352802367374$ & $0.07665846642009927$ \\
\hline
$\alpha_{78}$ & $0.04622400142944383$ & $0.08912467432180055$ \\
\hline
$\alpha_{79}$ & $0.06598393751625004$ & $0.06724339050226902$ \\
\hline
$\alpha_{80}$ & $0.015819026610491616$ & $0.11390203480637812$ \\
\hline
$\alpha_{81}$ & $0.014052365574156844$ & $0.1529879344816335$ \\
\hline
$\alpha_{82}$ & $0.019542717826361966$ & $0.09257293559305935$ \\
\hline
$\alpha_{83}$ & $0.02093163772726897$ & $0.13375170776745032$ \\
\hline
$\alpha_{84}$ & $0.03232182211334006$ & $0.10899217160505548$ \\
\hline
$\alpha_{85}$ & $0.035404672067686827$ & $0.08961421224461213$ \\
\hline
$\alpha_{86}$ & $0.04160032480693088$ & $0.0870469166593813$ \\
\hline
$\alpha_{87}$ & $0.03084632143248167$ & $0.11967303625257314$ \\
\hline
$\alpha_{88}$ & $0.03218274376106067$ & $0.08625153412085623$ \\
\hline
$\alpha_{89}$ & $0.027386520210324672$ & $0.11468071689788334$ \\
\hline
$\alpha_{90}$ & $0.0467579925718552$ & $0.09031490851523155$ \\
\hline
$\alpha_{91}$ & $0.03515363399072097$ & $0.06420968479797878$ \\
\hline
$\alpha_{92}$ & $0.009522308778970257$ & $0.08246536630622064$ \\
\hline
$\alpha_{93}$ & $0.050007623111272215$ & $0.06735253993260948$ \\
\hline
$\alpha_{94}$ & $0.027397549490475293$ & $0.07986056987421691$ \\
\hline
$\alpha_{95}$ & $0.040108142281991443$ & $0.08506428649843378$ \\
\hline
$\alpha_{96}$ & $0.04060265542768865$ & $0.06921061897885533$ \\
\hline
$\alpha_{97}$ & $0.06176115933187615$ & $0.07888370245488946$ \\
\hline
$\alpha_{98}$ & $0.05149748670123738$ & $0.0730839676106615$ \\
\hline
$\alpha_{99}$ & $0.030976848369531906$ & $0.07882644193751703$ \\
\hline
$\alpha_{100}$ & $0.04985378105030419$ & $0.07855811096717208$ \\
\hline
$\alpha_{101}$ & $0.02428257540185641$ & $0.0755618507268449$ \\
\hline
$\alpha_{102}$ & $0.039279772504672905$ & $0.06683328717340548$ \\
\hline
$\alpha_{103}$ & $0.018431969726226516$ & $0.07109645510485962$ \\
\hline
$\alpha_{104}$ & $0.01615117687134704$ & $0.07686292296039537$ \\
\hline
$\alpha_{105}$ & $0.033836619264623$ & $0.09207944256220246$ \\
\hline
$\alpha_{106}$ & $0.021684498478341585$ & $0.06792762522935986$ \\
\hline
$\alpha_{107}$ & $0.018653119555053665$ & $0.07184860065578008$ \\
\hline
$\alpha_{108}$ & $0.017510378838004492$ & $0.09077658256097626$ \\
\hline
$\alpha_{109}$ & $0.005027225774378641$ & $0.06892046751777886$ \\
\hline
$\alpha_{110}$ & $0.0050070660422215085$ & $0.08404266181153941$ \\
\hline
$\alpha_{111}$ & $0.008641122238781884$ & $0.05725657878308299$ \\
\hline
$\alpha_{112}$ & $0.0114109321956688$ & $0.04505359172704221$ \\
\hline
$\alpha_{113}$ & $0.00017017085816917188$ & $0.05865839976147119$ \\
\hline
$\alpha_{114}$ & $0.007227843412475732$ & $0.06098740030051164$ \\
\hline
$\alpha_{115}$ & $0.02380064289496081$ & $0.06750979580178162$ \\
\hline
$\alpha_{116}$ & $0.024626599428481333$ & $0.07232664164227215$ \\
\hline
$\alpha_{117}$ & $0.0002926203031912711$ & $0.07155889973262747$ \\
\hline
$\alpha_{118}$ & $0.00367483614722508$ & $0.07655628977344214$ \\
\hline
$\alpha_{119}$ & $0.003637542351726364$ & $0.08531209453810662$ \\
\hline
$\alpha_{120}$ & $0.0022174466541568516$ & $0.07272780537431511$ \\
\hline
$\alpha_{121}$ & $0.003972815375790473$ & $0.060692790181056167$ \\
\hline
$\alpha_{122}$ & $0.0063500940342546275$ & $0.07565018146829666$ \\
\hline
$\alpha_{123}$ & $0.0008190666659831924$ & $0.07435001036624961$ \\
\hline
$\alpha_{124}$ & $0.006404294461389681$ & $0.07641678559172299$ \\
\hline
$\alpha_{125}$ & $0.0772226658137164$ & $0.09172841413844901$ \\
\hline
$\alpha_{126}$ & $0.002848362891246903$ & $0.09045869915075516$ \\
\hline
$\alpha_{127}$ & $0.0012952627416890072$ & $0.05284222333171534$ \\
\hline
$\alpha_{128}$ & $0.017932379180303493$ & $0.07194325920411004$ \\
\hline
$\alpha_{129}$ & $0.007137167661640409$ & $0.08907570891638156$ \\
\hline
$\alpha_{130}$ & $0.03712900994359092$ & $0.09267691307775361$ \\
\hline
$\alpha_{131}$ & $0.0029178803264349185$ & $0.06180156823851851$ \\
\hline
$\alpha_{132}$ & $0.015565067465901694$ & $0.057769376722262844$ \\
\hline
$\alpha_{133}$ & $0.0007797083742386857$ & $0.06774002323306783$ \\
\hline
$\alpha_{134}$ & $0.045217214440781583$ & $0.0751076759531758$ \\
\hline
$\alpha_{135}$ & $0.0013741843692585687$ & $0.12059175834028163$ \\
\hline
$\alpha_{136}$ & $0.0003354018167419648$ & $0.08660859544741523$ \\
\hline
$\alpha_{137}$ & $0.0012121494697902024$ & $0.06185526343471609$ \\
\hline
$\alpha_{138}$ & $0.015325390110678683$ & $0.06456079230878453$ \\
\hline
$\alpha_{139}$ & $0.0028034548030816953$ & $0.0636821969541907$ \\
\hline
$\alpha_{140}$ & $0.0415339431984868$ & $0.07602483985713077$ \\
\hline
$\alpha_{141}$ & $0.002954384831987067$ & $0.08915221681102126$ \\
\hline
$\alpha_{142}$ & $0.028214095268082884$ & $0.0984722500891399$ \\
\hline
$\alpha_{143}$ & $0.008801691293012892$ & $0.09067271353727313$ \\
\hline
$\alpha_{144}$ & $0.011981667605959034$ & $0.09414865557456398$ \\
\hline
$\alpha_{145}$ & $0$ & $0.10168269428760995$ \\
\hline
$\alpha_{146}$ & $0.04442994587106425$ & $0.0909148528042305$ \\
\hline
$\alpha_{147}$ & $0.0025122969557108132$ & $0.09549983384551514$ \\
\hline
$\alpha_{148}$ & $0.005897723663266186$ & $0.07970401566114022$ \\
\hline
$\alpha_{149}$ & $0.0008298536197157702$ & $0.09550429166121593$ \\
\hline
$\alpha_{150}$ & $0.003146593473569992$ & $0.11367223296069545$ \\
\hline
$\alpha_{151}$ & $0.007423928474611485$ & $0.09621713402681015$ \\
\hline
\end{longtable}

\begin{longtable}{ |P{3cm}||P{5cm} ||P{5cm}| }
\caption{The values of $w$ for cases $2,3,\ldots,16$, used in the analysis of our algorithms for square and cube packing.}\label{w}\\
\hline
Case& Square packing & Cube packing\\
\hline
$2$ & $0.5218896004296165$ & $0.3559465695997889$ \\
\hline
$3$ & $0.6367683021976823$ & $0.3324106710303888$ \\
\hline
$4$ & $0.5508161595298383$ & $0.3547433890555143$ \\
\hline
$5$ & $0.6081996168574735$ & $0.29283548893321054$ \\
\hline
$6$ & $0.5966563767881228$ & $0.2680609843073525$ \\
\hline
$7$ & $0.5417242692011557$ & $0.30382397508342246$ \\
\hline
$8$ & $0.6988933681604961$ & $0.42984690908567424$ \\
\hline
$9$ & $0.7677036830017706$ & $0.7660334876156012$ \\
\hline
$10$ & $0.7691331237757477$ & $0.7674343307466625$ \\
\hline
$11$ & $0.773230983786544$ & $0.773273461291727$ \\
\hline
$12$ & $0.7836563381680435$ & $0.7932633383349649$ \\
\hline
$13$ & $0.7929071522802713$ & $0.8240834379579076$ \\
\hline
$14$ & $0.8113137810136913$ & $0.8470244201613557$ \\
\hline
$15$ & $0.8219971336489986$ & $0.88618415266251$ \\
\hline
$16$ & $0.872756492818088$ & $0.9152418129618586$ \\
\hline
\end{longtable}
\section{Counter examples}
\label{55}
In this section, we discuss the algorithm suggested by Han, Ye, and Zhou~\cite{HanYZ10} and its analysis, and show some deficiencies of that work.
We use the parameters of the journal version and in particular, we define counter-examples
for the bounds on the two-dimensional case claimed in that work. Our examples show that the asymptotic competitive ratio of the algorithm is higher than the claimed bound. This is shown not only for the analysis but for the actual output of the algorithm, found by its action on inputs suggested here, which is calculated by applying the algorithm for sufficiently large inputs and comparing the numbers of bins of the algorithm and of an offline solution.
We note that using our method of analysis, we can show an upper bound of approximately 2.14 on the asymptotic competitive ratio for the algorithm of ~\cite{HanYZ10}. Our examples show that the asymptotic competitive ratio cannot be smaller than $2.122$ while the claimed bound of \cite{HanYZ10} is $2.1187$.

We also show that the analysis of \cite{HanYZ10} cannot yield the claimed results, or any result that is not much larger, and moreover, the improvement over the bounds of \cite{epstein2005online} for the two-dimensional case is extremely small. Similar observations can be established  for earlier versions of this work \cite{HDZarxiv} and for the three-dimensional case (in both versions of this work \cite{HanYZ10,HDZarxiv}), by applying the weight functions provided in those papers.

Thus, we show that the results of \cite{HanYZ10} do not hold. We now explain the shortcomings of the proofs of this work. The analysis of \cite{HanYZ10,HDZarxiv} is based on a special case of the methods of \cite{S1}, using a partition into four cases. For cases 2 and 3, there are two weight functions, where any of these functions can be used for analysing optimal solutions. However, the proof of Lemma 6 of ~\cite{HanYZ10} is unclear and in fact the application of the method is incorrect. The correct way to apply the method is to select exactly one of the two functions for each case, and to use that function for testing all possible bins of an optimal solution. Obviously, one can choose the better function of the two in the sense that the maximum (or supremum) for every possible bin (of an optimal solution) is smaller. It is possible to find an upper bound (for the total weight of every possible bin) rather than finding the exact maximum or supremum for each function. The actual analysis of \cite{HanYZ10} for every case is split into five scenarios, such that one of the two functions of the case is chosen for each of the scenarios. However, we stress that the same function should have been chosen for all scenarios of one case. This is the reason that the analysis does not hold, and we can show counter-examples for the claimed asymptotic competitive ratio. The way to correct the analysis is to select one of the two functions for each case, and we show that no matter which of the two is used for case 2 of~\cite{HanYZ10}, this kind of analysis cannot yield an upper bound below $2.24069972$ for square packing, while the upper bound of \cite{epstein2005online} was approximately $2.244361$. We note that cases 1 and 4 of that work are correct, since a single weight function is proposed for each of these cases.




In order to discuss the algorithm presented in ~\cite{HanYZ10}, we present the parameters here, and we refer to the algorithm with these parameters simply as the algorithm of ~\cite{HanYZ10}. Table~\ref{han_intervals} consists of the interval partition, and the parameters used as $\alpha_i$ values. The algorithmic approach of the design of the algorithm is an adaptation of Super Harmonic algorithms into multiple dimensions, as we use here, so it is an Extended Harmonic algorithm (see Algorithm~\ref{EH}). The weight functions are mentioned later, when we show that the analysis cannot yield bounds close to the claimed ones.

\begin{table}[]
\begin{center}
 \begin{tabular}{c c c c c c c c}
 \hline
 $i$ & $(t_{i+1},t_i]$ & $\beta_i$ & $\delta_i$ & $\phi_i$ & $\gamma_i$ & $\theta_i=\beta_i^2-(\beta_i-\gamma_i)^2$ & $\alpha_i$ \ \\ [0.5ex]
 \hline
 1 & $(0.705, 1]$ & 1 & 0 & 0 & 0 & 0 & 0 \\
 2 & $(0.6475, 0.705]$ & 1 & 0.295 & 2 & 0 & 0  & 0\\
 3 & $(0.60, 0.6475]$ & 1 & 0.3525 & 3 & 0 &0  & 0\\
 4 & $(0.5, 0.60]$ & 1 & 0.4 & 4 & 0 &0  & 0\\
 5 & $(0.4, 0.5]$ & 2 & 0 & 0 & 0 &0  & 0\\
 6 & $(0.3525, 0.4]$ & 2 & 0.2 & 1 & 1 &3  & 0.1348\\
 7 & $(1/3, 0.3525]$ & 2 & 0.295 & 2 & 1 &3  & 0.2\\
 8 & $(0.295, 1/3]$ & 3 & 0 & 0 & 0 & 0  & 0\\
 9 & $(1/4, 0.295]$ & 3 & 0 & 0 & 1 & 5  & 0.3096\\
 10 & $(1/5, 1/4]$ & 4 & 0 & 0 & 1 & 7   & 0.2248 \\
 11 & $(1/6, 1/5]$ & 5 & 0 & 0 & 1 & 9  & 0.16\\
 12 & $(1/7, 1/6]$ & 6 & 0 & 0 & 1 & 11  & 0.13\\
 13 & $(1/8, 1/7]$ & 7 & 0 & 0 & 1 & 13  & 0.1\\
 14 & $(1/9, 1/8]$ & 8 & 0 & 0 & 1 & 15  & 0.1\\
 15 & $(0.1, 1/9]$ & 9 & 0 & 0 & 1 & 17  & 0.1\\
 16 & $(1/11, 0.1]$ & 10 & 0 & 0 & 2 & 36  & 0.05\\
 17 & $(0, 1/11]$ & $\ast$ & $\ast$ & $\ast$ & $\ast$ &  $\ast$ &  $\ast$ \\ [1ex]
\hline
\end{tabular}
\end{center}
\caption{\label{han_intervals} Intervals and other parameters used in the algorithm of \cite{HanYZ10}.}
\end{table}

Next, we define our example for case 2, and explain how the algorithm of ~\cite{HanYZ10} handles the example. We prove the following theorem using Lemmas \ref{1} and \ref{2}.

\begin{theorem}
Let $P_1$ and $P_2$ be the inputs defined below. Let the cost of the algorithm of \cite{HanYZ10} for $P_i$ be denoted by $\nu_i$.
It holds that $\frac{x_1}{OPT(P_1)}>y$ and $\frac{x_2}{OPT(P_2)}>y$, where $y=2.1187$ is the claimed upper bound on the asymptotic competitive ratio of that algorithm.
\end{theorem}
Let $\varepsilon>0$ be a sufficiently small positive constant (where in particular $\eps<0.00001$).

\paragraph{Input $\boldsymbol{P_1}$.}
Now, we will describe the first input $P_1$ (which can be defined to be arbitrarily large in the sense that the cost of an optimal solution can grow without bound), a feasible solution for the input, and the output of the algorithm for $P_1$.

Let $M, N$ be large positive integers (in this section the algorithm is fixed, and the roles of $M$ and $N$ in the algorithm are not used). In this input, we will require that $N\cdot(\frac{4\alpha_{12}}{11}+\frac{2\alpha_{9}}{5}) = M\cdot(1-\frac{2\alpha_{9}}{5}-\frac{2\alpha_{10}}{7}-\frac{5\alpha_{12}}{11})$ will hold. There are infinitely many positive integers $M$ such that $N$ is an integer as well, since all parameters are integral. Moreover, we can assume that $M$ and $N$ are both divisible by the required integers to make all numbers of bins discussed below integral too. This assumption will also be used for $P_2$ later.

The input is described by eight batches of items, arriving in the order defined below. The items of batch $j$ are called batch $j$ items.
   \begin{enumerate}

        \item $5M+4N $ items of size $\frac{1}{7}  + \varepsilon$,

        \item $2M$ items of size $\frac{1}{5}  + \varepsilon$,

        \item $2M + 2M$ items of size $\frac{1}{4} + \varepsilon$,

        \item $M$ items of size $\frac{1}{2}  + \varepsilon$,

        \item $N$ items of size $0.6 + \varepsilon$,

        \item $3M + 3N$ items of size $0.3525 + \varepsilon$,

        \item $24M + 25N$ items of size $\frac{1}{23} + \varepsilon$.

        \item Items of size $\eps$, whose total area is calculated later, such that these items are not packed into bins of the items of the previous batch.
    \end{enumerate}

\paragraph{A feasible Solution.}
It is possible to pack all of the items above in $M+N$ bins. See Figure \ref{OS1}. Obviously, the number of bins cannot be smaller than $M+N$, as this is the number of items whose sizes are larger than $\frac 12$, and no two such items can share a bin. In the solution, there are $M$ bins type A and $N$ bins type B. Every type A bin has $24$ items of batch 7, five items of batch 1, two items of batch 2, two items of batch 3, three items of batch 6, and one item of batch 4. Every type B bin has $25$ items of batch 7, four items of batch 1, two items of batch 3, three items of batch 6, and one item of batch 5.
    \begin{figure}[H]
	\includegraphics[width=1\columnwidth]{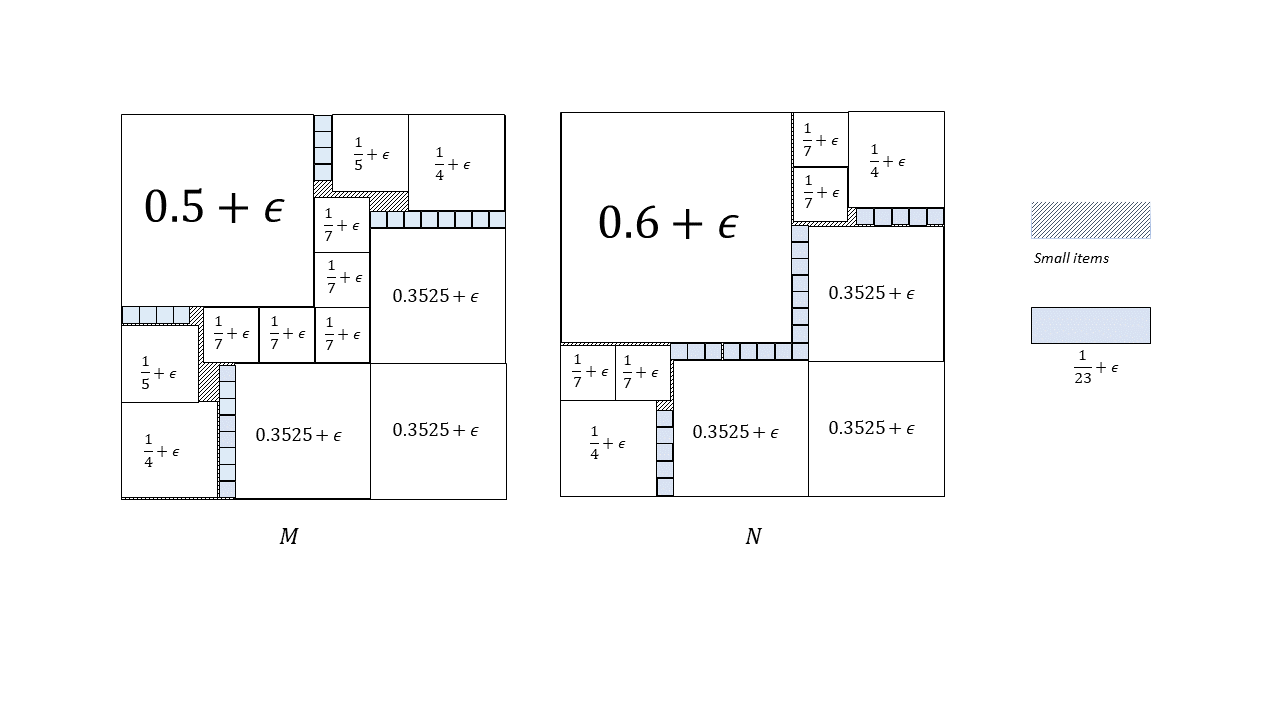}\vspace{-2cm}
	\caption{The two types of bins defined in an optimal solution for Input $P_1$, where type A is on the left hand size, and type B is on the right hand side.}
  	\label{OS1}
	\end{figure}

Based on the optimal solution, we find the area of items of size $\eps$.
The area of an item of size $\rho+\eps$ for $0<\rho \leq 0.65$ is below $\rho^2+2\eps$. Thus, neglecting terms that are linear (or quadratic) in $\eps$, and letting $\eps$ be chosen such that the bins of an optimal solution are filled completely, the area in a bin of type A is $1-24\cdot (\frac{1}{23})^2-5\cdot (\frac17)^2-2\cdot (\frac 15)^2-2\cdot (\frac 14)^2-3\cdot 0.3525^2-0.5^2=102944997/4147360000\approx 0.024821813635662$, and   the area in a bin of type B is $1-25\cdot (\frac{1}{23})^2-4\cdot (\frac17)^2-2\cdot (\frac 14)^2-3\cdot 0.3525^2-0.6^2=55324197/4147360000\approx 0.013339617732726$.
Note that the algorithm used for small items (of \cite{epstein2005optimal}) will not pack such items with items of side $\frac 1{23}+\eps$ into the same bins as long as the sizes are not in an interval of the form $(\frac {1}{10\cdot 2^k},\frac{1}{11\cdot 2^k}]$ for an integer $k \geq 0$, which can be avoided.

\begin{lemma}\label{1}
The number of bins for the output of  the algorithm of  \cite{HanYZ10} for $P_1$ is  approximately $11.4632218\cdot M$, and the asymptotic competitive ratio for $P_1$ is at least $2.12294632$.
\end{lemma}

\begin{proof}
We consider the input and the algorithm, and describe the packing performed for every batch. We use the parameters of the algorithm, based on the table. The line numbers below refer to Extended Harmonic, which is identical to the algorithm of \cite{HanYZ10}.
\begin{enumerate}
\item The first batch has $5M+4N$  items of type $12$.
Since the set of bins is currently empty, an application of EH results in new bins. Using lines $7$ and $17$ we find that there are $\frac{\alpha_{12}\cdot (4N+5M)}{11}$ bins  with eleven red items packed into each such bin, where these bins may receive blue items of some type later. Moreover, by lines $19$, $22$ and $25$, the algorithm will also open  $\frac{(1-\alpha_{12})(4N+5M)}{36}$ bins that have $36$ blue items, such that each bin can not receive any other item.

\item The second batch has $2M$  items of type $10$. The action is similar to the first batch. The items cannot be packed into previous bins since $\delta_{10}=\delta_{12}=0$.
The number of new bins with seven red items is $\frac{\alpha_{10}\cdot (2M)}{7}$, and the number of bins with $16$ blue items is  $\frac{(1-\alpha_{10})\cdot 2M}{16}$.

\item The third batch has $2N + 2M$ items of type $9$.  The action is similar to the first two batches, since $\delta_9=0$. The number of new bins with five red items is $\frac{\alpha_{9}\cdot (2M+2N)}{5}$, and the number of bins with nine blue items is  $\frac{(1-\alpha_{9})\cdot (2M+2N)}{9}$.

\item The fourth batch has $M$ items of type $4$. At this time, the algorithm has bins with blue items that cannot receive other items, but it has bins of type $(?,j)$  for $j \in {9,10,12}$, and every such bin can receive an item of type $4$. The number of such bins is $\frac{\alpha_{12}\cdot (4N+5M)}{11}+\frac{\alpha_{10}\cdot (2M)}{7}+\frac{\alpha_{9}\cdot (2M+2N)}{5}$. We required that $N\cdot(\frac{4\alpha_{12}}{11}+\frac{2\alpha_{9}}{5}) = M\cdot(1-\frac{2\alpha_{9}}{5}-\frac{2\alpha_{10}}{7}-\frac{5\alpha_{12}}{11})$, and therefore  $\frac{\alpha_{12}\cdot (4N+5M)}{11}+\frac{\alpha_{10}\cdot (2M)}{7}+\frac{\alpha_{9}\cdot (2M+2N)}{5}=M$. Every item of type $4$ is added to a bin of red items, and
as a result, no bin that can receive other items remains after all items of the fourth batch are presented.

\item The fifth batch has $N$ items of type $3$.
Since $\alpha_{4}=0$, and there are no open bins, the algorithm will open $N$ bins, such that each one of the bins has exactly one item of size $0.6$ + $\varepsilon$. These bins could potentially receive red items later, but it will not be able to receive red items of type $6$. There are no further items of types $7,8,\ldots,16$, and thus these bins will not receive other items.

\item The sixth batch has $3N+3M$ items of type $6$. Since no previous bins can receive further items, the packing is as for the three first batches, and there are new bins with three red items, where the number of these bins is $\frac{\alpha_{6}\cdot (3M+3N)}{3}$, and the number of bins with four blue items is  $\frac{(1-\alpha_{6})\cdot (3M+3N)}{4}$.

\item The seventh batch consists of $24M+25N$ items of size $\frac{1}{23}+ \varepsilon$.
Since the size of the items is less than $\frac{1}{11}$, the items will be pack using algorithm \textit{AssignSmall}.
Every bin can have $22^2$ items, and hence, the algorithm will open new ($24M + 25N)\cdot\frac{1}{22^2}$ bins, such that each bin contains $22^2$ items of size $\frac{1}{23}  + \varepsilon$.

	\item The eighth batch consists of items of size $\varepsilon$. By our assumption, they are packed by \textit{AssignSmall} into new bins. The number of bins is approximately the total area of these items, which is $\frac {102944997 \cdot M+ 55324197 \cdot N}{4147360000}$.
\end{enumerate}

To find the cost of the algorithm, and by using the requirement $N\cdot(\frac{4\alpha_{12}}{11}+\frac{2\alpha_{9}}{5}) = M\cdot(1-\frac{2\alpha_{9}}{5}-\frac{2\alpha_{10}}{7}-\frac{5\alpha_{12}}{11})$, we get that $N=
\frac{724609}{164696} \cdot M$, where $\frac{724609}{164696}\approx 4.399675766$. Thus, the optimal cost as a function of $M$ is $
\frac{889305}{164696} \cdot M$.

Using the analysis of \cite{HanYZ10}, it can be noted that after packing all the items in the input, the only open bins for red items are $(?,6)$, which means there are no bins of type $(?,j)$ for $j>6$ and the input belongs to case $2$ of that analysis.

Based on the output, we will get the following total cost:
$$\frac{(1-\alpha_{12})(4N+5M)}{36}+\frac{(1-\alpha_{10})\cdot (2M)}{16}+\frac{(1-\alpha_{9})\cdot (2M+2N)}{9}+\frac{\alpha_{6}\cdot (3M+3N)}{3}$$ $$+\frac{(1-\alpha_{6})\cdot (3M+3N)}{4}+\frac{24M + 25N}{22^2}+\frac {102944997 \cdot M+ 55324197 \cdot N}{4147360000}+M+N \ .$$

By using the values of $\alpha_6$, $\alpha_9$, $\alpha_{10}$, and $\alpha_{12}$, we find that the total cost is approximately $11.4632218067166\cdot M$ and the resulting lower bound on the asymptotic competitive ratio of the algorithm is $2.12294632176699$.
\end{proof}

\paragraph{Input $\boldsymbol{P_2}$.}
Now, we will describe the second input.
Let $M, N$ be large positive integers.
In this input, we will require that $N\cdot\frac{2\alpha_{9}}{5} = M\cdot(1-\frac{5\alpha_{12}}{11}-\frac{2\alpha_{10}}{7}-\frac{2\alpha_{9}}{5})$ will hold.

The input is described by ten batches of items, arriving in the order defined below.
   \begin{enumerate}

        \item $M$ items of size $\frac{1}{2}  + \varepsilon$,

        \item $5M$ items of size $\frac{1}{7}  + \varepsilon$,

        \item $2M$ items of size $\frac{1}{5} + \varepsilon$,

        \item $2N + 2M$ items of size $\frac{1}{4}  + \varepsilon$,

        \item $3N + 3M$ items of size $\frac{1}{3} + \varepsilon$,

        \item $N$ items of size $0.6475 + \varepsilon$,
    	
        \item $8M + 8N$ items of size $\frac{1}{13} + \varepsilon$,

        \item $6N$ items of size $\frac{1}{12} + \varepsilon$,

        \item $10M$ items of size $\frac{1}{22} + \varepsilon$.

        \item Items of size $\eps$, whose total area is calculated later, such that these items are not packed into bins of the items of the previous batches.

    \end{enumerate}

 \begin{figure}[H]
 \vspace{-0.9cm}
    \includegraphics[width=1\columnwidth]{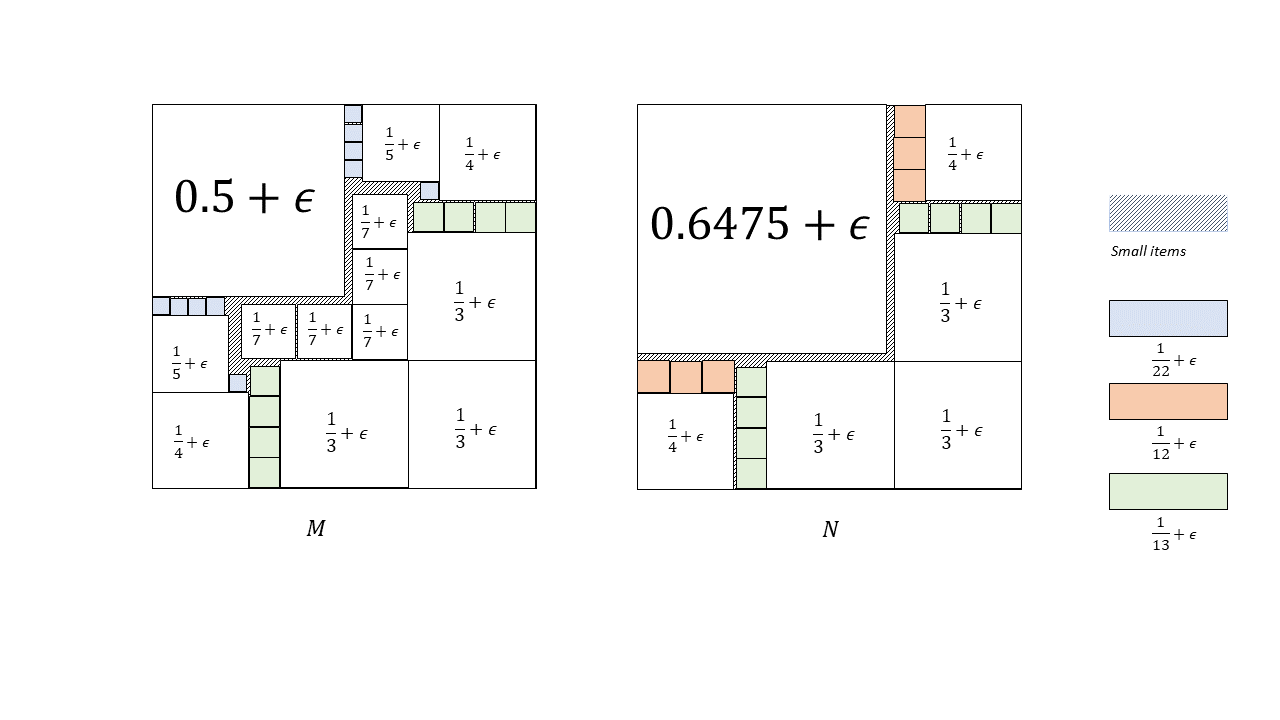}\vspace{-2cm}
	\caption{The two types of bins defined in an optimal solution for Input $P_2$. Type A is on the left hand size, and type B is on the right hand side.}
  	\label{OS2} \vspace{-0.4cm}
	\end{figure}	

\paragraph{A feasible Solution.}
    It is possible to pack all of the items above in $M+N$ bins. See Figure~\ref{OS2}. Obviously, the number of bins cannot be smaller than $M+N$, since this is the number
of items whose size is larger that $\frac{1}{2}$. In the solution, there are $M$ bins type A and $N$ bins type B. Each bin of type A has
has ten items of batch $10$, eight items of batch $7$, three items of batch $5$, two items of batch $4$, two items of batch $3$, five items
of batch $2$, and one item of batch $1$. Every bin of type B has six items of batch $8$, eight items of batch $7$, one item of batch $6$, three items of batch $5$ and two items of batch $4$.

Based on the optimal solution, we find the area of items of size $\eps$.
Similarly to previous example, we let $\eps$ tend to zero, and fill all the bins of the optimal solution completely.
The area in a bin of type A is
$1-10\cdot (\frac{1}{22})^2-8\cdot (\frac{1}{13})^2-5\cdot (\frac17)^2-2\cdot (\frac 15)^2-2\cdot (\frac 14)^2-3\cdot (\frac 13)^2-0.5^2=25026427/601200600\approx 0.041627415208834$,

and   the area in a bin of type B is $1-8\cdot (\frac{1}{13})^2-6\cdot (\frac{1}{12})^2-2\cdot (\frac 14)^2-3\cdot (\frac 13)^2-0.6475^2=903311/27040000\approx 0.033406471893491$.
Note that the algorithm used for small items (of \cite{epstein2005optimal}) will not pack such items with items of sides $\frac 1{22}+\eps, \frac{1}{13}+\eps, \frac 1{12}\eps$ into the same bins as long as the sizes are not in an interval of the forms $(\frac {1}{(i+1)\cdot 2^k},\frac{1}{i\cdot 2^k}]$, for an integer $k \geq 0$, and $i=21, 12, 11$, which can be avoided.

\begin{lemma}\label{2}
The number of bins for the output of the algorithm of~\cite{HanYZ10} for $P_2$ is  approximately $15.01872658\cdot M$, and the asymptotic competitive ratio for $P_1$ is at least $2.120087899$.
\end{lemma}

\begin{proof}

We consider the input and the algorithm, and describe the packing performed for every batch. We use the parameters of the algorithm, based on the table.
	  \begin{enumerate}
	\item	The first batch has $M$ items of type $4$. Since the set of bins is currently empty, an application of EH results in $M$ new bins, with one blue item of type $4$ in each such bin, where these bins could potentially receive red items later.
\item The second, third and fourth batches include $5M$ items of size $\frac{1}{7}+\varepsilon$, $2M$ items of size $\frac{1}{5}+\varepsilon$ and $2N+2M$ items of size $\frac{1}{4}+\varepsilon$. An application of EH results in coloring $\alpha_{12}\cdot5M,\alpha_{10}\cdot2M$, and $\alpha_{9}\cdot(2M+2N)$ items (respectively) as red. The number of bins needed to pack all these red items is $\frac{\alpha_{12}\cdot(5M)}{11}+\frac{\alpha_{10}\cdot(2M)}{7}+\frac{\alpha_9\cdot(2N+2M)}{5}$. We required that $N\cdot\frac{2\alpha_{9}}{5} = M\cdot(1-\frac{5\alpha_{12}}{11}-\frac{2\alpha_{10}}{7}-\frac{2\alpha_{9}}{5})$ will hold, and therefore $\frac{\alpha_{12}\cdot5M}{11}+\frac{\alpha_{10}\cdot2M}{7}+\frac{\alpha_9\cdot(2N+2M)}{5} = M$. At this time, the algorithm has $M$ open bins of type $(4,?)$. All the red items will be packed into these open bins. Moreover, by lines $19$, $22$ and $25$, the algorithm will also open $\frac{(1-\alpha_{12})(5M)}{36}$ ,$\frac{(1-\alpha_{10})(2M)}{16}$ ,$\frac{(1-\alpha_{9})(2N+2M)}{9}$ bins that have $36,16,9$ blue items respectively, such that each bin can not receive any other item.

\item The fifth batch has $3N+3M$ items of type $7$. Since no previous bins can receive further items, using lines $7$ and $17$ of Algorithm~\ref{EH} we find that this results in opening $\frac{\alpha_{7}\cdot(3N+3M)}{3}$ bins with three red items packed into each such bin where these bins could potentially receive blue items later, but the further blue items belong to a type that is too large. Moreover, by lines $19$ and $28$, the algorithm will also open $\frac{(1-\alpha_{7})(3N+3M)}{4}$ bins that have four blue items, such that each bin can not receive any other item.

\item The sixth batch has $N$ items of size $0.6475+\varepsilon$. Since $\alpha_2=0$, and the only open red bins are type $(?,7)$ which can not receive blue items of type $2$, the Algorithm will open $N$ bins such that each bin has one blue item of type $2$. These bins could potentially receive red items later, but there no more red items in the input.

\item The seventh, eighth and ninth batches include $8N+8M$ items of size $\frac{1}{13}+\varepsilon$, $6N$ items of size $\frac{1}{12}+\varepsilon$ and $10N$ items of size $\frac{1}{22}+\eps$ respectively. Since the size of these items is less than $\frac{1}{11}$, the items will be packed using algorithm \emph{AssignSmall}. Each one of the sizes belongs to a different type of small items, and it is packed independently of other sizes.
    Hence, the algorithm will open $\frac{8N+8M}{12^2}+\frac{6N}{11^2}+\frac{10M}{21^2}$ new bins.

\item The tenth batch consists of items of size $\varepsilon$. By our assumption, they are packed by \textit{AssignSmall} into new bins. The number of bins is approximately the total area of these items, which is $\frac{25026427}{601200600}\cdot M+ \frac{903311}{27040000}\cdot N$.
\end{enumerate}
	
To find the cost of the algorithm, and by using the requirement $N\cdot\frac{2\alpha_{9}}{5} = M\cdot(1-\frac{5\alpha_{12}}{11}-\frac{2\alpha_{10}}{7}-\frac{2\alpha_{9}}{5})$, we get that $N=
\frac{724609}{119196}  \cdot M$, where $\frac{724609}{119196}\approx 6.079138561696701$. Thus, the optimal cost as a function of $M$ is $
\frac{843805}{119196} \cdot M$.

Using the analysis of \cite{HanYZ10}, it can be noted that after packing all the items in the input, the only open bins for red items are $(?,7)$, which means there are no bins of type $(?,j)$ for $j>7$ and the input belongs to case $3$ of that analysis.

Based on the output, we will get the following total cost:
$$ M+N+\frac{(1-\alpha_{12})(5M)}{36}+\frac{(1-\alpha_{10})\cdot (2M)}{16}+\frac{(1-\alpha_{9})\cdot (2M+2N)}{9}+\frac{\alpha_{7}\cdot (3M+3N)}{3}$$ $$+\frac{(1-\alpha_{7})\cdot (3M+3N)}{4}
+\frac{8N+8M}{12^2}+\frac{6N}{11^2}+\frac{10M}{21^2}+
\frac{25026427}{601200600}\cdot M+ \frac{903311}{27040000}\cdot N \ .$$

By using the values of $\alpha_6$, $\alpha_9$, $\alpha_{10}$, and $\alpha_{12}$, we find that the total cost is approximately $15.018726578408019\cdot M$ and the resulting lower bound on the asymptotic competitive ratio of the algorithm is $2.120087899087498$.

%
%
%
%
%
%
%
%
\end{proof}

Next, we discuss the two weight functions of \cite{HanYZ10} for case 2 and square packing. All input items except for small items will be of types $3, 4, 6, 9, 10, 12$, so we only define the weights of such items according to the two weight functions. For small items, the weight of an item is defined to be $1.2$ times its area (for both weight functions). The first weight function is called $W_{2,1}$, and its is defined that the weights of items of these types are: $1, 0, \frac{1-\alpha_6}4+\frac{\alpha_6}3, \frac{1-\alpha_9}9+\frac{\alpha_9}5, \frac{1-\alpha_{10}}{16}+\frac{\alpha_{10}}7, \frac{1-\alpha_{12}}{36}+\frac{\alpha_{12}}{11}$, respectively. The second weight function is called $W_{2,2}$, and its is defined that the weights of items of these types are: $1, 1, \frac{1-\alpha_6}4+\frac{\alpha_6}3, \frac{1-\alpha_9}9, \frac{1-\alpha_{10}}{16}, \frac{1-\alpha_{12}}{36}$, respectively.

We will discuss one bin (which may be a bin of an optimal solution) whose weight is high for $W_{2,1}$ and another bin whose weight is high for $W_{2,2}$. Thus, one cannot choose one of the functions are use it to prove an upper bound below $2.2$ for case 2 (and thus for the entire algorithm).

The first bin is identical to the type B bin in the optimal solution for $P_2$. The total area of small items is approximately (letting $\eps$ tend to zero) $1-0.6^2-3\cdot 0.3525^2-2\cdot (1/4)^2-4\cdot (1/7)^2\approx =475093/7840000\approx 0.060598596938776$. Thus, the weight according to $W_{2,1}$ is $1+3\cdot(\frac{1-\alpha_6}4+\frac{\alpha_6}3)+2\cdot (\frac{1-\alpha_9}9+\frac{\alpha_9}5)+4\cdot (\frac{1-\alpha_{12}}{36}+\frac{\alpha_{12}}{11}) +1.2\cdot (475093/7840000)\approx 2.277619932488147$.
The second bin is similar to the type A bin in the optimal solution for $P_1$, but instead of the items of type 12 there are just small items. The total area of small items is approximately $1-0.5^2-3\cdot 0.3525^2-2\cdot (1/4)^2-2\cdot (1/5)^2=0.17223125$. Thus, the weight according to $W_{2,2}$ is $1+3\cdot(\frac{1-\alpha_6}4+\frac{\alpha_6}3)+2\cdot (\frac{1-\alpha_9}9)+2\cdot (\frac{1-\alpha_{10}}{16}) +1.2\cdot 0.17223125\approx 2.240699722$.

For cube packing, correcting the analysis of \cite{HanYZ10} would still give improved bounds, though these bounds would be closer to $2.7$ than to $2.6$.

The last example is not original, and it does not deal with the algorithm of  \cite{HanYZ10} but with all Extended Harmonic algorithms and similar algorithms. It was known for a while that algorithms that use types and do not combine items based on their exact sizes cannot have an asymptotic competitive ratio below $1.5833333$ for one dimension. A similar construction for squares and cubes can be found in \cite{HeyS17}. These constructions consist of a large number of inputs, and the calculations of \cite{HeyS17} are not always justified mathematically (inequalities are used as equalities without any explanation). However, the results in fact hold, and we provide a short proof for that.

\begin{proposition}
Every Extended Harmonic algorithm, for any dimension $d \geq 1$, has an asymptotic competitive ratio of at least $3-\frac 1{2^d}-\frac 1{4^d}-\frac{2^{d+1}}{3^d}+\frac{2}{3^d}$. In particular, for $d=1,2,3$, the lower bounds on the asymptotic competitive ratios of such algorithms are approximately $1.5833333$, $2.0208333$, and $2.34085648$.
\end{proposition}
\begin{proof}
We will introduce two inputs, and these inputs will consist of four types of items. There are small items, and items of size $\frac 12+\delta$ for a very small $\delta>0$. Consider the type that contains the values $\frac 13$ and $\frac 23$. Consider the type for $\frac 13$. If it is the right endpoint of the interval for the type, we move to the next type (whose left endpoint is $\frac 13$). For example, if the interval is $(\frac 14,\frac 13]$, we use the next one (and its form is $(\frac 13, t_j]$), but if the interval contains $\frac 13$ as an internal point (for example, $(0.33, 0.34]$), we use that interval. Since intervals are half open and half closed, and have positive lengths, the interval of $\frac 23$ does not have it as a left endpoint.

Let $\eps$ be a sufficiently small value that in particular satisfies the property  that $\frac 13+\eps$ and $\frac 23-\eps$ are interval points of the considered intervals (for example, if the first one is $(\frac 13,0.336]$ and the second one is $(0.666,\frac 23]$, we can use $\eps=0.001$). Let $\beta$ be the proportion of red items for the type of $\frac 13+\eps$.
Let $N$ be a large positive integer.

The first input also has $(2^d-1)\cdot N$ items of size $ \frac 13+\eps$,  and it has $N$ items of size $\frac 12+\eps$. It has small items with total volume of $N\cdot(1-\frac{2^d-1}{3^d}-\frac{1}{2^d})$. An optimal solution has $N$ bins, each with $2^d-1$ items of size $\frac 13+\eps$, one item of size $\frac 12+\eps$ and small items, and there is one additional bin with small items (for a sufficiently small value of $\eps$), where this bin can be neglected for large values of $N$. The algorithm has $N$ bins with items of size $\frac 12+\eps$, and since the small items are packed separately and the number of bins is at least their volume, the number of bins for these items is at least $N\cdot(1-\frac{2^d-1}{3^d}-\frac{1}{2^d})$.
There are $\frac{\beta \cdot (2^d-1)\cdot N}{2^d-1}=\beta \cdot N$ bins with red items of size $\frac 13+\eps$, and  $\frac{(1-\beta) \cdot (2^d-1)\cdot N}{2^d}=N \cdot (1-\beta)\cdot (1-\frac 1{2^d})    $ bins with blue items.
Since all or some of the red items can be possibly packed with items of size $\frac 12+\eps$, we do not take the bins with red items into account (though their existence could possibly increase the cost of the algorithm).
No other items can be combined.
Thus, the cost of the algorithm is at least $N\cdot (1+(1-\frac{2^d-1}{3^d}-\frac{1}{2^d})+\frac{(1-\beta)(2^d-1)}{2^d})=N\cdot(3-\frac 1{2^{d-1}}+\frac{1}{3^d}-\frac{2^d}{3^d}-\beta \cdot (1-\frac 1{2^d}))   $, and we get $3+\frac{1}{3^d}-\frac 1{2^{d-1}}-\frac{2^d}{3^d}-\beta \cdot (1-\frac1{2^d})\leq R$.

The second input has $(2^d-1)\cdot N$ items of size $ \frac 13+\eps$ and $N$ items of size $\frac 23-\eps$. It has small items with total volume of $N\cdot(1-\frac{2^d-1}{3^d}-\frac{2^d}{3^d})=N\cdot (1-\frac{2^{d+1}-1}{3^d})$. Note that this amount is non-negative for all integers $d\geq 1$, and it is equal to zero for $d=1$.
An optimal solution for this input has $N$ bins with one item of the larger size, and $2^d-1$ items of the smaller size. In addition it has small items in every bin (and there might be one bin of small items only).
However, since the interval of $\frac 23-\eps$ has $\frac 23$ in its interval, the algorithm cannot combine items of the two sizes into a bin (since the right endpoints of the two types are too large). The algorithm creates (up to a constant number of bins, due to rounding) $\frac{\beta\cdot (2^d-1)\cdot N}{2^d-1}=\beta\cdot N$ bins with $2^d-1$ red items of size $\frac 13+\eps$, and $\frac{(1-\beta) \cdot (2^d-1)\cdot N}{2^d}=N\cdot (1-\beta )\cdot (1-\frac 1{2^d})$ bins with sets of $2^d$ blue items. It also creates $N$ bins with items of size $\frac 23-\eps$, and $N\cdot(1-\frac{2^{d+1}-1}{3^d})$ bins of small items.
The total number of bins is $N\cdot (\beta+1-\beta-\frac 1{2^d}+\frac{\beta}{2^d}+1+(1-\frac{2^{d+1}-1}{3^d}))=N\cdot(3-\frac{2^{d+1}}{3^d}+\frac{1}{3^d}-\frac 1{2^d}+ \frac{\beta}{2^d})$, and by letting $R$ be the asymptotic competitive ratio, we get $3-\frac{2^{d+1}}{3^d}+\frac{1}{3^d}-\frac 1{2^d}+ \frac{\beta}{2^d} \leq R$. By multiplying this inequality by $2^d-1$ we get
$3\cdot 2^d-\frac{2^{2d+1}}{3^d}+\frac{2^d}{3^d}-4+\frac{2^{d+1}}{3^d}-\frac{1}{3^d}+\frac 1{2^d} +\beta \cdot (1-\frac1{2^d})\leq (2^d-1)\cdot R$.

Taking the sum of the two inequalities we have  $(3+\frac{1}{3^d}-\frac 1{2^{d-1}}-\frac{2^d}{3^d}-\beta \cdot (1-\frac1{2^d}))+(3\cdot 2^d-\frac{2^{2d+1}}{3^d}+\frac{2^d}{3^d}-4+\frac{2^{d+1}}{3^d}-\frac{1}{3^d}+\frac 1{2^d} +\beta \cdot (1-\frac1{2^d}))\leq 2^d \cdot R$.
By rearranging, we get $2^d \cdot R \geq 3\cdot 2^d-\frac{2^{2d+1}}{3^d}+\frac{2^{d+1}}{3^d}-1-\frac 1{2^d}$, and therefore $R \geq 3-\frac{2^{d+1}}{3^d}+\frac{2}{3^d}-\frac1{2^d}-\frac 1{4^d}$.
\end{proof}

\bibliographystyle{plain}

\appendix

\section{An example for Algorithm EH}
In what follows, we show how EH behaves for a given example.
We now provide the intervals parameters $\alpha_i$ and $\Delta_i$ which are required by Algorithm~\ref{EH};
see Table~\ref{example}.

\begin{table}[H]
\centering
\begin{tabular}{ P{1cm}P{1.5cm}P{1cm}P{1cm}P{1cm}P{1cm}  }
 \hline
 i& $(t_{i+1},t_ {i}]$	&$\Delta_i$&$\phi_i$&$\gamma_i$&$\alpha_i$\\
 \hline
1  & $(0.7,1]$	&$0$&$0$&$0$&$1$\\
2  &	($\frac{2}{3}$,$0.7$] & $0.3$&$1$&$0$&$1$\\
3  &($\frac{1}{2},\frac{2}{3}$]	& $\frac{1}{3}$&$2$&$0$&$1$\\
4  &($\frac{1}{3},\frac{1}{2}$]	& $0$&$0$&$0$&$1$\\
5  &	($0.3,\frac{1}{3}$]]& $0$ &$0$&$1$&$0.4$\\
6  &	(0.1,0.3]	& $0$&$0$&$1$&$0.4$\\
7  &	(0,0.1]	&*&*&*&*\\
  \hline
 \end{tabular}
\caption{Input parameters for Algorithm~\ref{EH}.}\label{example}
\end{table}
Note that in any face of the bin, item of type 3 has at least $\frac{1}{3}$ space left so any red item with
size at most $\frac{1}{3}$ can be packed in a $(3,?)$ bin. Similarly, any red item with size at most $0.3$ can be packed in a $(2,?)$ bin.
\begin{table}[H]
\centering
\begin{tabular}{ P{2cm}P{1cm}P{4cm}  }
 \hline
$j = \phi(i)$& $\Delta_j$	&Red items accepted\\
 \hline
1	& $0.3$	& $6$\\
2  &	$\frac{1}{3}$ & $5,6$\\
  \hline
 \end{tabular}
\end{table}
First, we will describe the input. The input is described by $5$ batches of items, arriving in the order defined below.

\paragraph{Input $\boldsymbol{I}$.}

\begin{enumerate}

        \item one item of size $0.9$,

        \item $2$ items of size $\frac{2}{3}  $,

        \item $2$ items of size $0.3$,

        \item $14$ items of size $\frac{1}{3}$,

        \item $12$ items of size $0.3$,

    \end{enumerate}
In what follows, upon receiving multiple items, we will illustrate how they are packed using Algorithm~\ref{EH}.
First, items $0.9,\frac{2}{3},\frac{2}{3},0.3,0.3$ which fall in intervals $I_1, I_3, I_6,I_6$ respectively, arrive. Since $e_i \geq \lfloor \alpha_is_i\rfloor$ for $i=1,3,6$ in each step and $\phi(1),\phi(6)=0$, all the items will be colored blue and four new bins of types $(1),(3,?),(3,?),(6)$ will be opened to pack these items. We call them $(a),(b),(c),(d)$ respectively; See Figure \ref{ex1}.

 \begin{figure}[H]
	 \includegraphics[width=0.7\columnwidth,center]{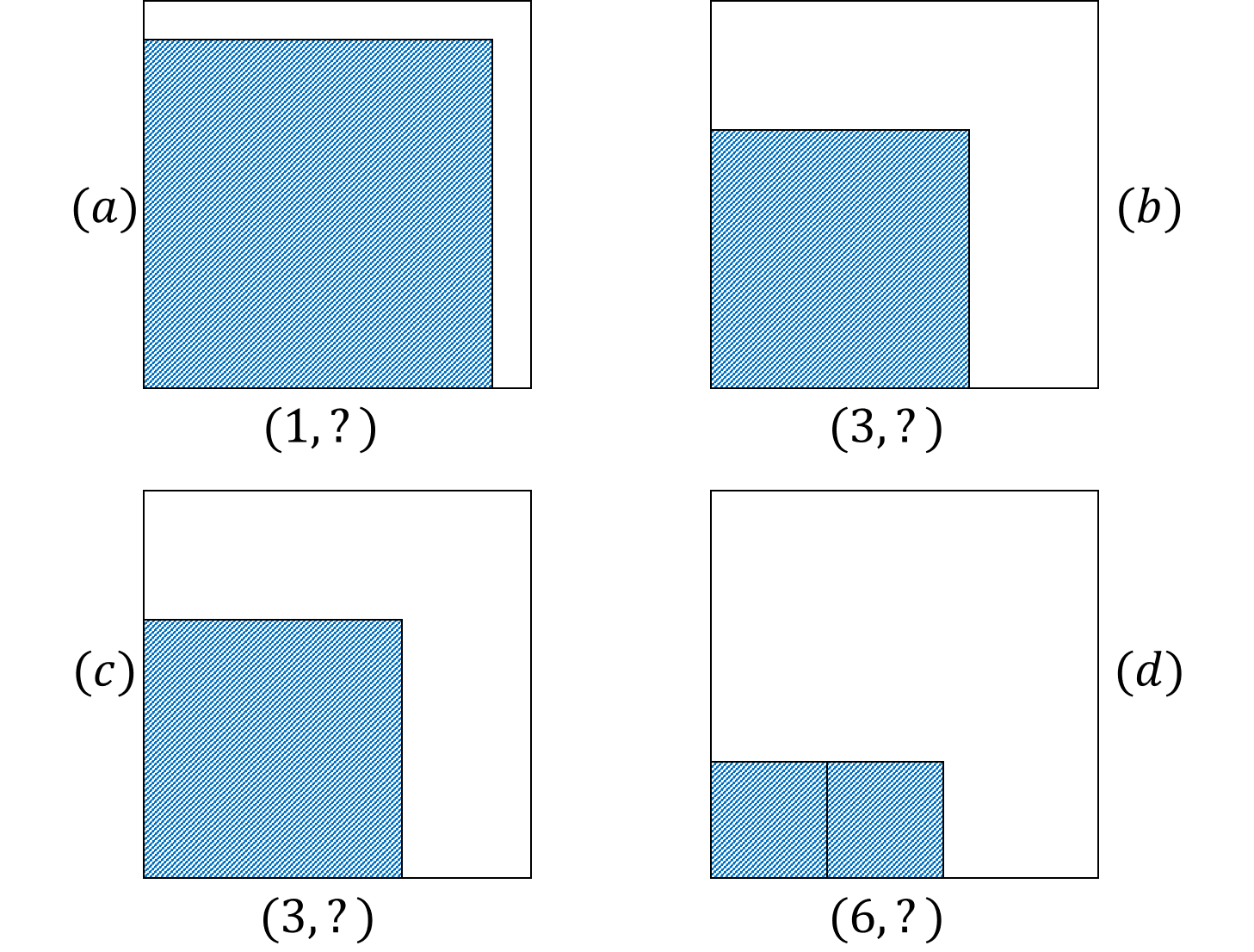}
\caption{}
\label{ex1}
\end{figure}
Note that $\delta_1 = 0$ and $\beta_1=1$, i.e., there is no space left in bins type $(1)$ for red items and only one
item from interval $1$ can be packed in a bin of type $(1)$. Hence, bin $(a)$ is considered \say{closed}; see Figure \ref{ex2} for \say{active} bins.
 \begin{figure}[H]
	 \includegraphics[width=0.7\columnwidth,center]{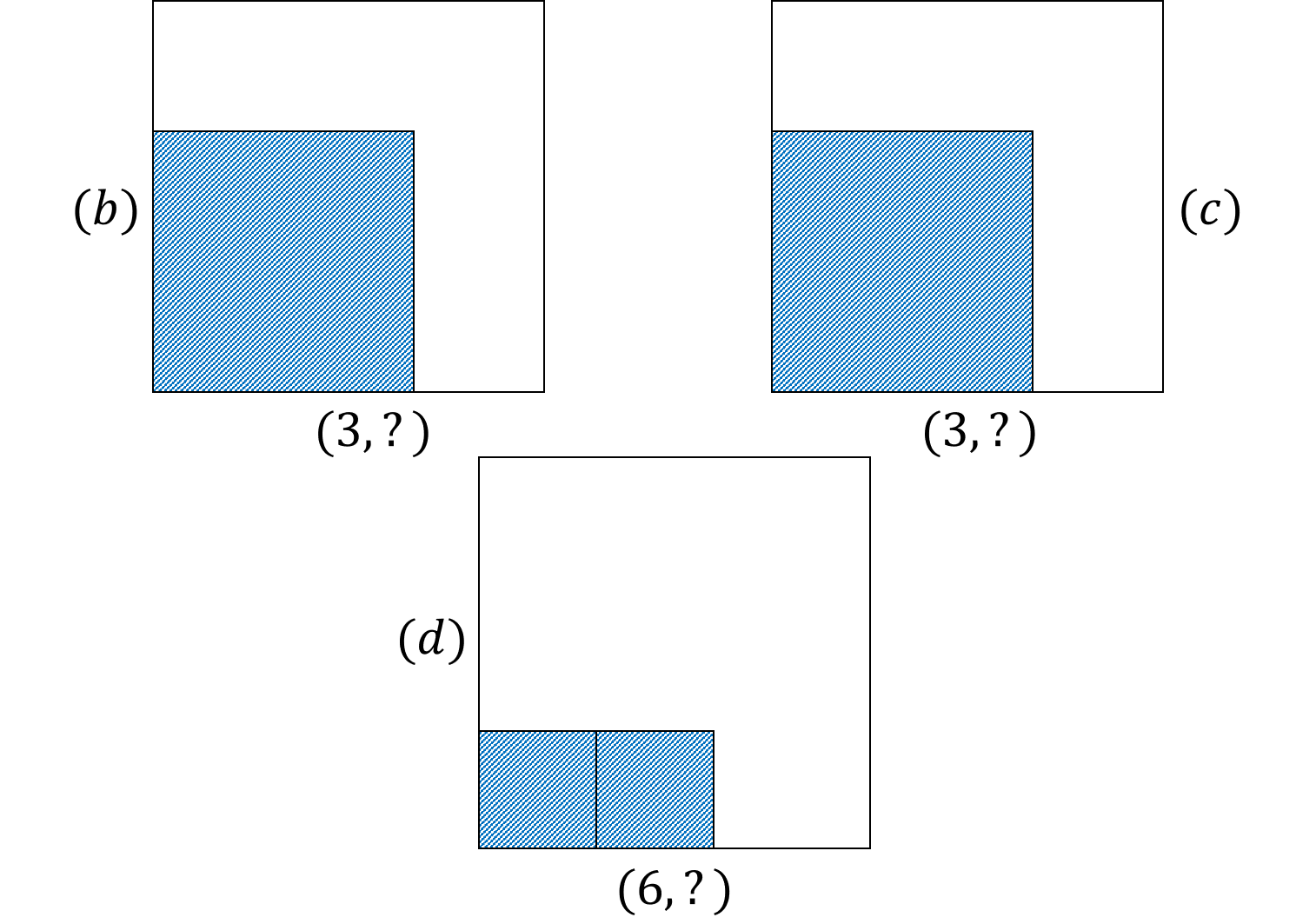}
\caption{}
\label{ex2}
\end{figure}
The next fourteen items in the input sequence are of type $5$. Since $\alpha_5=0.4$, five of the fourteen items will be colored red and the remaining nine items will be colored blue.
The red items will be packed in the \say{active} bin (b) since $\Delta_{\phi(5)} < \gamma_3 t_3$, and all the nine blue items will be packed in a new bin which called $(e)$; see Figure \ref{ex3}.
 \begin{figure}[H]
	 \includegraphics[width=0.7\columnwidth,center]{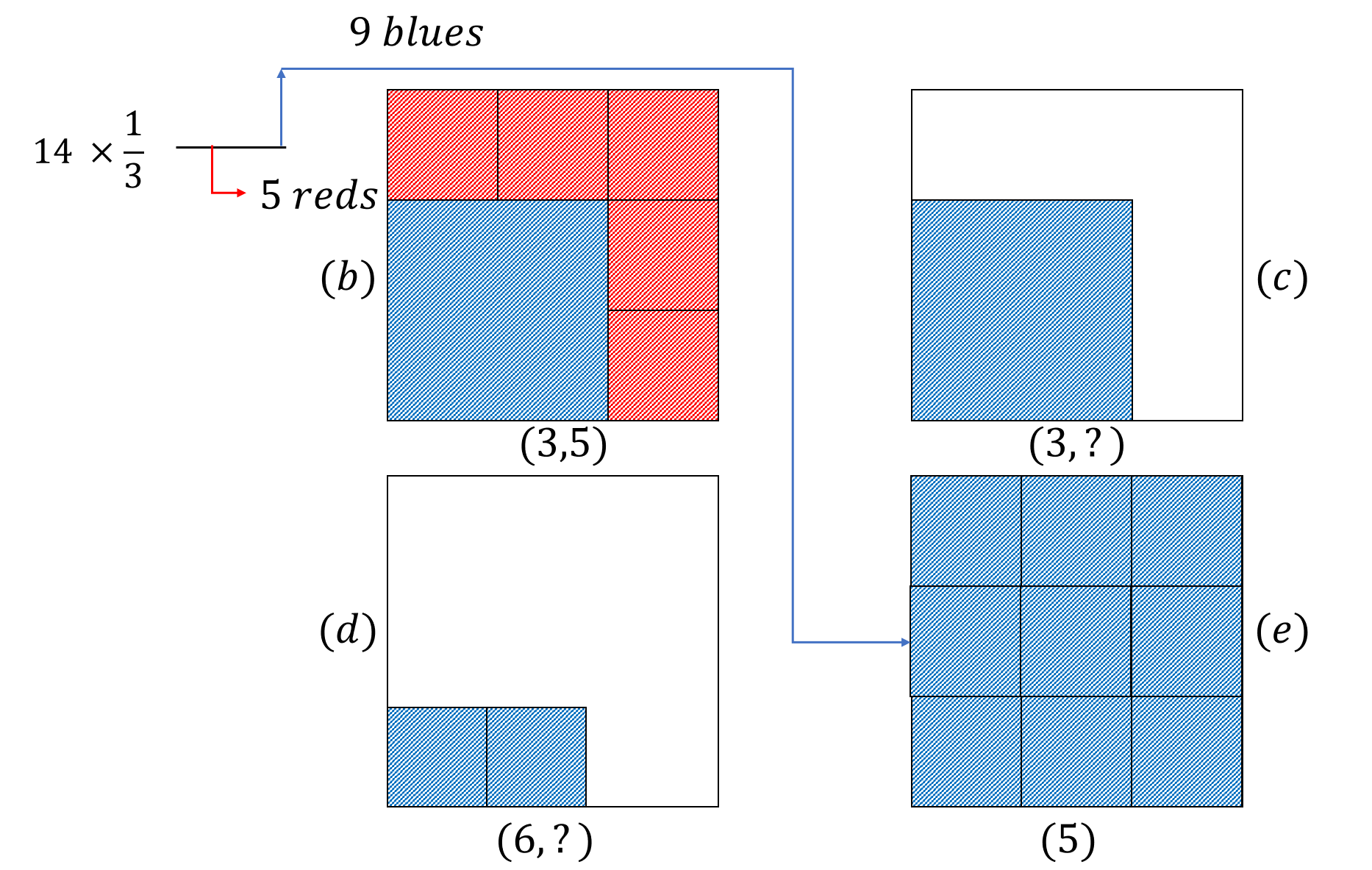}
\caption{}
\label{ex3}
\end{figure}
Note that bins $(b),(e)$ are now considered \say{closed} since they have no space left. The next twelve items in the input sequence are of type $6$, and observe that there have been already two items of type $6$ in the input. Since
$\alpha_6=0.4$, five of the twelve items will be colored red while the remaining seven items will be colored blue.
The red items will be packed in bin (b) since $\Delta_{\phi(5)} < \gamma_3 t_3$, while the remaining blue items (of the twelve items) will be packed in the \say{active} bin of type $(6)$; see Figure \ref{ex3}.

 \begin{figure}[H]
	 \includegraphics[width=0.7\columnwidth,center]{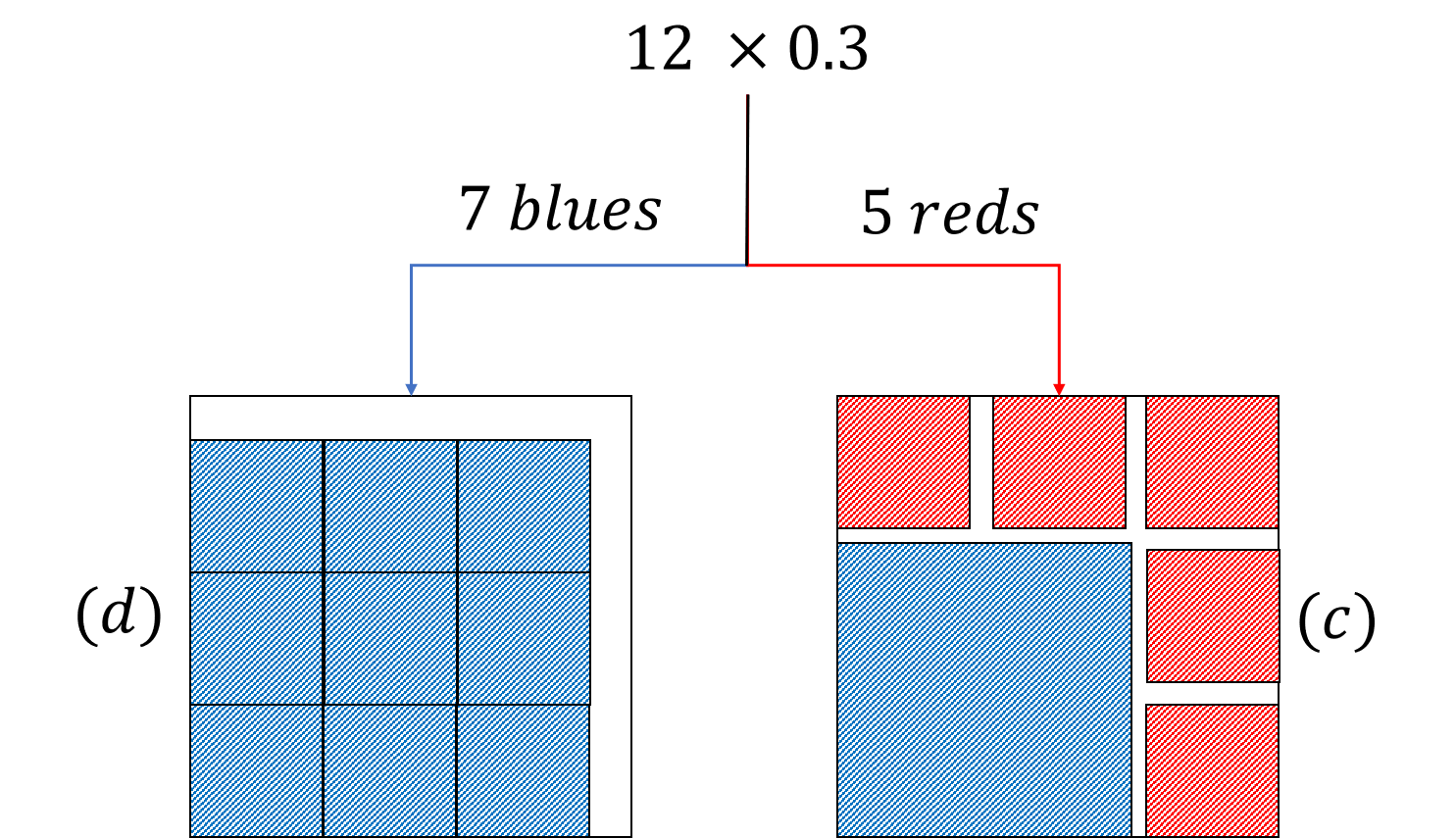}
\caption{The packing for the last $12$ items.}
\label{ex4}
\end{figure}
The output of Algorithm~\ref{EH} on the aforementioned example, is illustrated below.
 \begin{figure}[H]
	 \includegraphics[width=0.7\columnwidth,center]{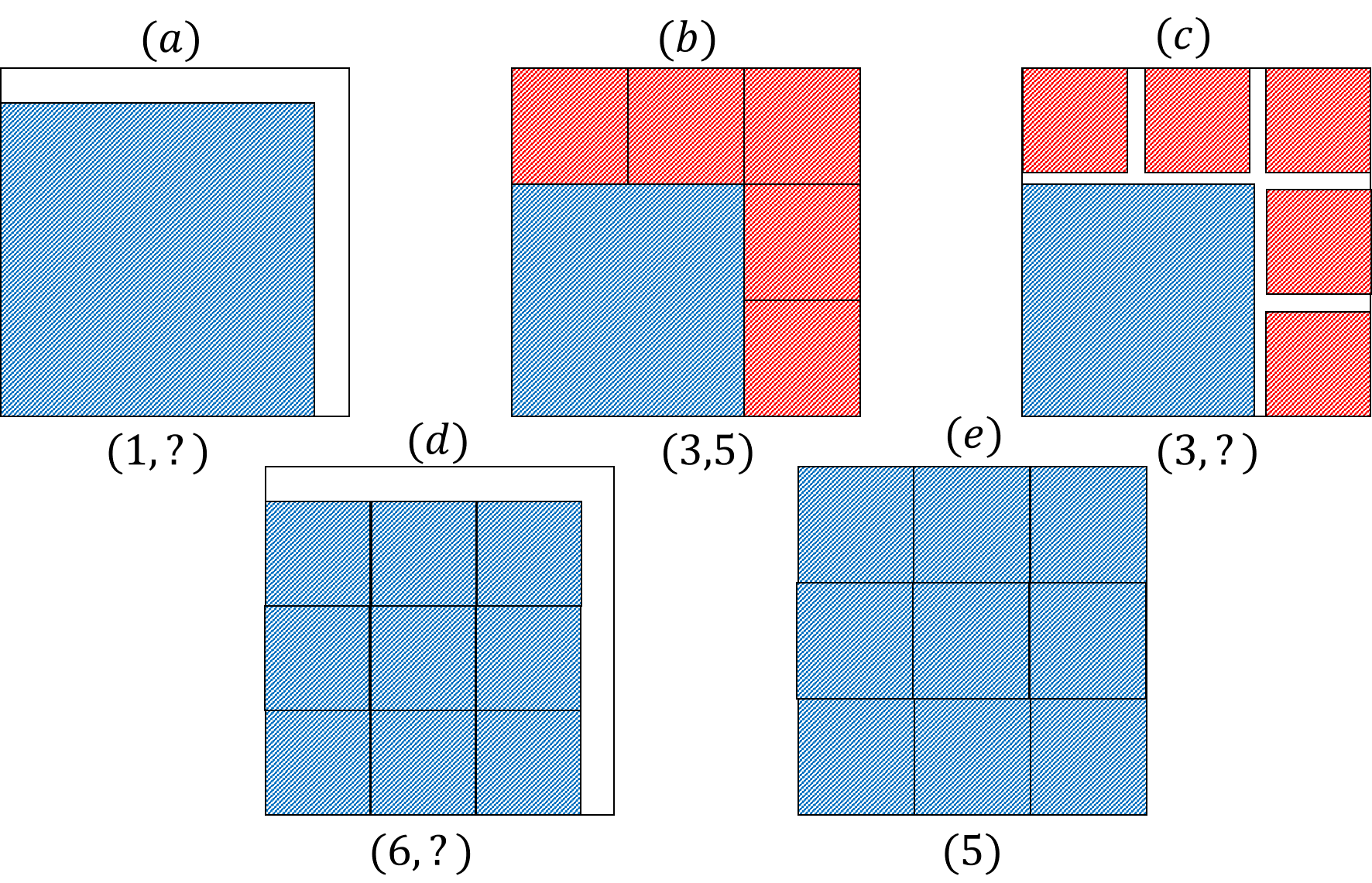}
\caption{The final result of the algorithm.}
\label{ex5}
\end{figure}

\end{document}